\documentclass[a4paper,UKenglish]{lipics-v2016}
 
\usepackage{microtype}

\author[1]{Ivor Hoog v.d.}
\author[2]{Elena Khramtcova}
\author[1]{Maarten L\"{o}ffler.}
\affil[1]{Dept.\ of Inform.\ and Computing Sciences, Utrecht University, the Netherlands\\
  \texttt{[i.d.vanderhoog|m.loffler]@uu.nl}}
\affil[2]{Computer Science Department, Universit\'e libre de Bruxelles (ULB), Belgium\\
  \texttt{elena.khramtsova@gmail.com}}
\authorrunning{I. Hoog v.d., E. Khramtcova, M.  L\"{o}ffler} 

\Copyright{Ivor Hoog v.d., Elena Khramtcova, and Maarten L\"{o}ffler}


\EventEditors{John Q. Open and Joan R. Access}
\EventNoEds{2}
\EventLongTitle{42nd Conference on Very Important Topics (CVIT 2016)}
\EventShortTitle{CVIT 2016}
\EventAcronym{CVIT}
\EventYear{2016}
\EventDate{December 24--27, 2016}
\EventLocation{Little Whinging, United Kingdom}
\EventLogo{}
\SeriesVolume{42}
\ArticleNo{23}

\usepackage[utf8]{inputenc}
\usepackage{amsmath}
\usepackage{amsfonts}
\usepackage{amsthm}
\usepackage{graphicx}
\usepackage{wrapfig}
\usepackage{algorithm}
\usepackage[noend]{algpseudocode}
\usepackage[colorinlistoftodos]{todonotes}
\usepackage{makecell}
\usepackage{xspace}
\usepackage{caption}

\newtheorem{observation}{Observation}

\newcommand{\R}{\ensuremath{\mathbb R}}
\newcommand{\etal}{\textit{et al.}\xspace}

\newcommand{\mycomment}[4]{\textsc{#1 #3:} \textcolor{#2}{\textsl{#4}}}

\newcommand{\maarten}[2][says]{\mycomment{Maarten}{olive}{#1}{#2}}

\renewcommand{\mycomment}[4]{}

\newcommand{\thmheadfont}{\textcolor{darkgray}{$\blacktriangleright$}\nobreakspace\sffamily\bfseries}
\newenvironment{repeatenv}[2]%
  {\vspace{1em}\noindent{\thmheadfont #1~\ref{#2}.}\ \slshape}
  {\normalfont\vspace{1em}}

\begin{document}

\title{Dynamic smooth compressed quadtrees Fullversion.}
\maketitle

\begin {abstract}
  We introduce dynamic smooth (a.k.a. balanced) compressed quadtrees with worst-case constant time updates in constant dimensions.
  We distinguish two versions of the problem.
  First, we show that quadtrees as a space-division data structure can be made smooth and dynamic subject to {\em split} and {\em merge} operations on the quadtree cells.
  Second, we show that quadtrees used to store a set of points in $\R^d$ can be made smooth and dynamic subject to {\em insertions} and {\em deletions} of points.
  The second version uses the first but must additionally deal with {\em compression} and {\em alignment} of quadtree components.
  In both cases our updates take $2^{\mathcal{O}(d\log d )}$ time, except for the point location part in the second version which has a lower bound of $\Theta (\log n)$---but if a pointer ({\em finger}) to the correct quadtree cell is given, the rest of the updates take worst-case constant time.
  Our result implies that several classic and recent results (ranging from ray tracing to planar point location) in computational geometry which use quadtrees can deal with arbitrary point sets on a real RAM pointer machine.
\end {abstract}

\section {Introduction}

The quadtree is a hierarchical spacial subdivision data structure based on the following scheme: starting with a single square, iteratively pick a square and subdivide it into four equal-size smaller squares, until a desired criterion is reached.
Quadtrees and their higher-dimensional equivalents have been long studied in computational geometry~\cite {Vaidya88, CallahanKo95, KrznaricLe98, BernEpTe99, EppsteinGoSu08, BuchinMu11, LM12, BY17}
and are popular among practitioners because of their ease of implementation and good performance in many practical applications~\cite {FinkelBe74, macdonald1990heuristics, BernEpGi94, aronov2006cost, de2012kinetic, bennett2016planar}.
We will not review the extremely rich literature here and instead refer to the excellent books by Samet~\cite {Samet90} and Har-Peled~\cite {H11}.
\maarten {Check if Samet's earlier 1984 survey \cite {Samet84} is completely contained in his book or should be cited separately.}

\paragraph* {Smooth and Dynamic Quadtrees}

A quadtree is {\em smooth}\footnote {Also called {\em balanced} by some authors, which is not to be confused with the notion of balance in trees related to the relative weights of subtrees.} if its each leaf is comparable in size to adjacent leaves.
It has been long recognized that smooth quadtrees are useful in many applications~\cite {BernEpGi94},
and smooth quadtrees can be computed in linear time (and have linear complexity) from their non-smooth counterparts~\cite[Theorem~14.4]{de2008computational}.

A quadtree is {\em dynamic} if it supports making changes to the structure in sublinear time.
Recently, quadtrees have been applied in kinetic and/or uncertain settings that call for dynamic behaviour of the decomposition~\cite {de2012kinetic, lss13, khramtcova2017dynamic}.
Bennett and Yap~\cite {BY17} show how to maintain a smooth {\em and} dynamic quadtree subject to amortized constant-time {\em split} and {\em merge} operations on the quadtree leaves.

At this point, it is useful to distinguish between the quadtree {\it an sich}, a combinatorial subdivision of space, and the quadtree as a data structure for storing a point set (or other set of geometric objects).
Given a set $P$ of points in the plane and a square that contains them, we can define the minimal quadtree that contains $P$ to be the quadtree we obtain by recursively subdividing the root square until no leaf contains more than one (or a constant number of) point(s). It is well-known that such a minimal quadtree can have superlinear complexity, but can still be stored in linear space by using {\em compression}~\cite {de2008computational}. Additionally, when working in the Real RAM computation model, it may not be possible to keep different compressed components properly aligned~\cite {H11, LM12}. These complications imply that we cannot simply apply results known for standard/regular (henceforth called {\em uncompressed}) quadtrees. When maintaining a dynamic quadtree storing a point set $P$, we wish to support high-level operations of inserting points into $P$ and removing points from $P$.\footnote {In Table~\ref {tab:operations} in Section~\ref {sec:preliminaries} we provide a complete list of operations and how they relate.}

\paragraph* {Contribution}

In this paper, we show that it is possible to maintain a quadtree storing a set of points $P$ that is smooth and possibly compressed, which supports worst-case constant-time insertions and deletions of points into $P$, assuming we are given a pointer ({\em finger}) to the cell of the current quadtree containing the operation.
Our result runs on a Real RAM and generalises to arbitrary constant dimensions.

In the first half of the paper (Sections~\ref{sec:static}-\ref{sec:dimensions}), we focus on the problem of making the quadtree itself dynamic and smooth, improving the recent result by Bennett and Yap~\cite {BY17} from amortized to worst-case constant time split and merge operations.
The challenge here is to avoid cascading chains of updates required to maintain smoothness.
Our key idea is to introduce several layers of smoothness: we maintain a core quadtree which is required to be $2$-smooth, but cells added to satisfy this condition themselves need only be $4$-smooth, etc (refer to Section~\ref {sec:preliminaries} for the formal definition of $2^j$-smoothness).
In Section~\ref{sec:static}, we show that when defining layers in this way, we actually need only two layers in $\R^2$, and the second layer will always already be smooth.
In Section~\ref{sec:dynamic}, we show that we can handle updates on the core quadtree in constant time.
In Section~\ref{sec:dimensions}, we generalise the result to arbitrary dimensions (now, the number of layers, and thus the smoothness of the final tree, depends on the dimension).

In the second half of the paper (Sections~\ref{sec:compression}-\ref{sec:alignment}), we focus on lifting our result to quadtrees that store a set of points on a pure real-valued pointer machine.
The challenge here is to redefine compressed quadtrees in a consistent way across different layers of smoothness, and to re-align possibly misaligned components on the fly when such components threaten to merge.
In Section~\ref{sec:preliminaries}, we show that we can view insertions of points as two-step procedures, where we first need to locate the correct leaf of the current tree containing the new point, and then actually insert it into the quadtree. We show in Section~\ref{sec:compression} that we can still handle the second step in worst-case constant time.
In Section~\ref{sec:alignment}, we deal with the issue of avoiding the use of the floor operation, which is not available on a pure Real RAM.

\paragraph* {Implications}

In many applications of quadtrees, it is useful to be able to walk in constant time from any leaf of the tree to any of its neighboring leafs. Clearly, in a non-smooth quadtree, a single leaf may have many smaller neighbors; in some applications, however, it is enough to only maintain pointers from every leaf to every {\em larger} neighboring leaf.
However, it has been shown that {\em dynamically} maintaining such pointers is not possible unless the quadtree is smooth~\cite {BY17}.\footnote{See Section~\ref{sec:applications} for a formal definition and a proof.}

A large number of papers in the literature explicitly or implicitly rely on the ability to efficiently navigate a quadtree, and our results readily imply improved bounds from amortized to worst-case~\cite {macdonald1990heuristics, khramtcova2017dynamic, lss13, de2012kinetic, park2012self}, and extends results from bounded-spread point sets to arbitrary point sets~\cite {
devillers1992fully}. Other papers could be extended to work for dynamic input with our dynamic quadtree implementation~\cite{aronov2006cost, LM12, mezger2003hierarchical}. Several dynamic applications are in graphics-related fields and are trivially paralellizable, which enhanced the need for worst-case bounds. In Section~\ref {sec:applications} we give a comprehensive overview of the implications of our result.

\section{Preliminaries. }
\label{sec:preliminaries}

In this section we review several necessary definitions. 
Well-known and existing concepts are underligned.
Consider the $d$-dimensional real space $\mathbb{R}^d$. For a hypercube $R \subset \mathbb{R}^d$, 
the \textbf{size} of $R$, denoted $|R|$, is the length of a $1$-dimensional facet (i.e., an edge) of $R$.

\begin{definition}[\underline{Quadtree}]
Let $R$ be an axis-aligned hypercube in $\mathbb{R}^d$. A \textbf{quadtree} $T$ on the root cell $R$ is a hierarchical decomposition
of $R$ into smaller axis-aligned hypercubes called \textbf{quadtree cells}. Each node $v$ of $T$ has an associated
quadtree cell $C_v$, and $v$ is either a leaf or it has $2^d$ equal-sized children whose cells subdivide $C_v$.\footnote {We follow~\cite {LM12} in using {\em quadtree} in any dimension rather than dimension-specific terms (i.e. {\em octree}, etc).} 
\end{definition}

From now on, unless explicitly stated otherwise, 
when talking about a quadtree cell $C$ we will be meaning both $C$ and the quadtree node corresponding to $C$.

\begin{definition}[\underline{Neighbor}, Sibling neighbor]
\label{def:neighbor}
Let $C$ and $C'$ be two cells of a quadtree $T$ in $\mathbb{R}^d$. 
We call $C$ and $C'$ \textbf{neighbors}, 
if they are interior-disjoint and share (part of) a $(d-1)$-dimensional facet. 
We call $C$ and $C'$ \textbf{sibling neighbors} if they are neighbors and they have the same parent cell.
\end{definition}

\begin{observation}
Let $C$ be a quadtree cell. Then: (i)  $C$ has at most $2d$ neighbors of size $|C|$; and  
(ii) For each of the $d$ dimensions, $C$ has exactly one sibling neighbor that neighbors $C$ in that dimension. 
\end{observation}

\begin{definition}[\underline{$2^j$-smooth cell, $2^j$-smooth quadtree}]
For an integer constant $j$, we call a cell $C$ \textbf{$2^j$-smooth} if 
the size of each leaf neighboring $C$  is at most $2^j|C|$.
If every cell in a quadtree is $2^j$-smooth, the quadtree is called
\textbf{$2^j$-smooth}.
\end{definition}

\begin{observation}
If all the quadtree leaves are $2^j$-smooth, then all the intermediate cells 
are $2^j$-smooth as well.\footnote{Observe that if a single (leaf) cell is $2^j$-smooth, its ancestors do not necessarily have to be such.}
 That is, the quadtree is $2^j$-smooth.
\end{observation}

\begin{definition}[Family related]
Let $C_1, C_2$ be two cells in a quadtree $T$ such that $|C_1| \le |C_2|$. If the parent of $C_2$ is an ancestor of $C_1$ we call $C_1$ and $C_2$ \textbf{family related}.\footnote{Observe that $C_1$ and $C_2$ do not have to be neighbors.} 
\end{definition}

We now consider quadtrees that store point sets. 
Given a set $P$ of points in  $\mathbb{R}^d$, and a hypercube $R$ containing all points in $P$, 
 an \textbf{uncompressed quadtree} that stores $P$ is a quadtree $T$ in   $\mathbb{R}^d$ on the root cell $R$, that can be obtained by starting from $R$ and 
 successively subdividing every cell that contains at least two points in $P$ into $2^d$ child cells.  

\begin{definition}[\underline{Compression}]
Given a large constant $\alpha$, an $\alpha$-compressed quadtree is a quadtree with additional \textbf{compressed nodes}. A compressed node $C_a$
has only one child $C$ with $|C| \leq |C_a|/\alpha$, and the region $C_a \setminus C \subset \mathbb{R}^d$ does not contain any points in $P$.
We call the link between $C_a$ and $C$ a \textbf{compressed link}, and $C_a$ the \textbf{parent} of the compressed link. 
\end{definition}

Compressed nodes induce a partition of a compressed quadtree $T$ into a collection of uncompressed quadtrees interconnected by compressed links. 
We call the members of such a collection the 
\textbf{uncompressed components} of $T$.

\begin{table}[t]
\begin{tabular}{ | l | l  |}

\hline
  \textcolor{gray}{\textbf{Operation}}  & \textcolor{gray}{\textbf{ Running time}} \\ \hline
\textbf{I. Quadtree operations (uncompressed quadtree)} & \\ \hline
    Split a cell 
& $\mathcal{O}((2d)^d)$ \\ \hline
 Merge cells 
& $\mathcal{O}((2d)^d)$ \\ \hline
  \textbf{ II. Quadtree operations ($\alpha$-compressed quadtree)} & \\ \hline
 Insert a 
component  
& $\mathcal{O}(d^2(6d)^d)$ \\ \hline
 Delete a component & $\mathcal{O}(d^2(6d)^d)$\\ \hline
 Upgrowing of a component & $\mathcal{O}(\log(\alpha)d^2(6d)^d)$ \\ \hline
Downgrowing of a component & $\mathcal{O}(\log(\alpha)d^2(6d)^d)$\\   
    \hline

  \textbf{ III. Operations on the point set $P$, stored in a quadtree} &  \\ \hline
  Insert a point into $P$ & $\mathcal{O}(d\log(n) + \log(\alpha)d^2(6d)^d)$ \\ \hline
  Insert a point into $P$, given a finger & $\mathcal{O}(\log(\alpha)d^2(6d)^d)$ \\ \hline
  Delete a point from $P$ & $\mathcal{O}(\log(\alpha)d^2(6d)^d)$ \\ \hline
  \end{tabular}
\caption{Operations considered in this paper and the running times of the provided implementation.}\vspace{-2em} 
\label{tab:operations}
\end{table}

\subsection{Quadtree operations and queries}
Table~\ref {tab:operations} gives an overview of quadtree operations.
It is insightful to distinguish three levels of operations.
Operations on a compressed quadtree (II) internally perform operations on an uncompressed quadtree (I).
Similarly, operations on a point set stored in a quadtree (III) perform operations on the quadtree.
We now give the formal definitions and more details. 

\begin{definition}[\underline{split, merge}] 
\label{def:split-merge}
 Given a leaf cell $C$ of a quadtree $T$, the \textbf{split} operation for $C$ inserts the $2^d$ equal-sized children of $C$ into $T$.  
Given  a set $2^d$ leaves of a quadtree $T$ whose parent is the same cell, the \textbf{merge} operation these $2^d$ cells.  
\end{definition}

\begin{definition}[upgrowing, downgrowing]
\label{def:up-down}
Let $A$ be an uncompressed component  of a compressed quadtree.  
\textbf{Upgrowing} of $A$ adds the parent $R'$ of the root cell $R$ of $A$. Cell $R'$ becomes the root of component $A$. 
\textbf{Downgrowing} of $A$ removes the root $R$ of $A$, and all the children of $R$ except one child $C$. Cell $C$ becomes the root of $A$. 
The downgrowing operation requires $R$ to be an internal cell, and all the points stored in $A$ to be contained in  one child $C$ of $R$. 
\end{definition}

An insertion of a point $p$ into the set $P$ stored in an  $\alpha$-compressed quadtree $T$ is performed in two phases: first, the leaf cell of $T$ should be 
found that contains 
$p$; second, the quadtree should be updated.
The first phase, called \textbf{point location}, can be performed in $\mathcal{O}(d\log(n))$ time (See edge oracle trees in \cite{lss13}), and can be considered a query in our data structure. 
 We refer to the second phase separately as \textbf{inserting a point given a finger}, see Table~\ref{tab:operations}.

\section{Static non-compressed smooth quadtrees in $\mathbb{R}^1$ and $\mathbb{R}^2$.}
\label{sec:static}

We first view the quadtree as a standalone data structure subject only to merge and split operations. In this section we are given a unique non-smooth, uncompressed quadtree $T_1$ over $\mathbb{R}^1$ or $\mathbb{R}^2$ with $n$ cells. 
It is known \cite[Theorem 14.4]{de2008computational} that an uncompressed quadtree can be made smooth by adding 
$\mathcal{O}(n)$ cells. However, the reader can imagine that if we want all the cells to be $2$-smooth that we cannot make the quadtree dynamic with worst-case constant updates because balancing keeps cascading (Appendix \ref{sec:bennet}). In this section we show that if $T_1$ is a quadtree over $\mathbb{R}^1$ or $\mathbb{R}^2$ then we can
 extend $T_1$ by consecutively adding $d \in \{1,2\}$\footnote{In Section~\ref {sec:dimensions} we show the same is possible for quadtrees of arbitrary dimension $d$.} sets of cells,
i. e., cells of $d$ different \emph{brands}, 
so that in the resulting \emph{extended} quadtree $T^*$ each cell is smooth according to its brand. 
The total number of added cells is $\mathcal{O}(d\cdot 2^d n)$. 

\subsection{Defining our smooth quadtree.}

We want to add a minimal number of cells to the original quadtree $T_1$ such that the cells of $T_1$ become $2$-smooth and the balancing cells are smooth with a constant dependent on $d \in \{1,2\}$. In general we want to create an \textbf{extended quadtree} $T^*$ with $T_1 \subset T^*$ where all cells with brand $j$ are $2^j$-smooth for $j \le d+1$.

\begin{figure}[H]
	\centering
	\includegraphics[width=300px]{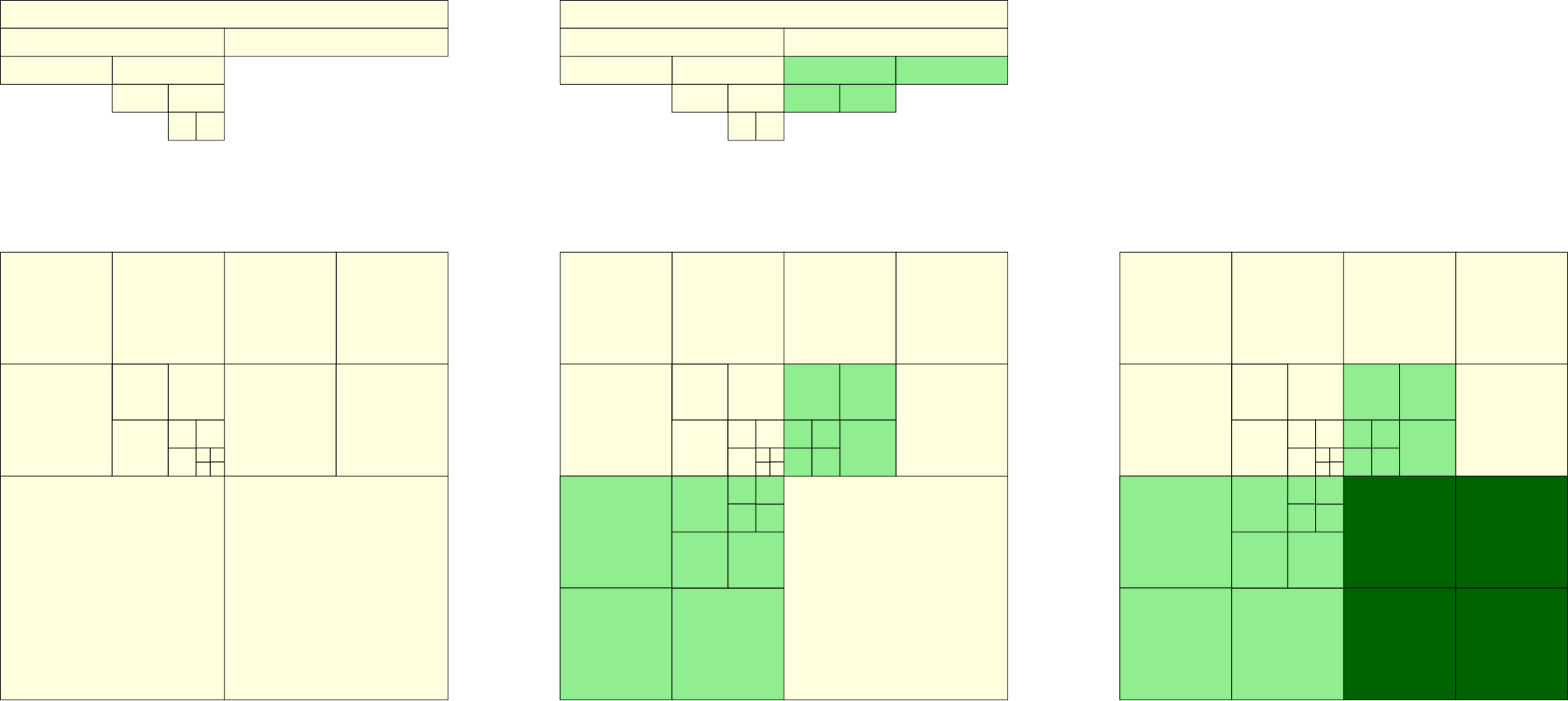}
	\caption{
Left: a quadtree in $\mathbb{R}^1$ (up) and $\mathbb{R}^2$ (down); Center: 
the (light-green) cells of brand 2 added; Right: the (dark-green) cells of brand 3 added. 
In each row, the rightmost tree is the smooth version of the leftmost one.
}
	\label{fig:balanceexample}
\end{figure}

The \emph{true} cells  ($T_1$) get brand $1$. Figure \ref{fig:balanceexample} shows two 
quadtrees and its balancing cells. This example 
also illustrates our main result:  to balance a tree $T_1$ over $\mathbb{R}^d$ we use $(d+1)$ different types of cells and 
the cells of the highest brand 
are automatically $2^{d+1}$-smooth.
This example gives rise to an intuitive, recursive definition for balancing cells in $\mathbb{R}^d$. In 
this definition we have a slight abuse of notation: For each brand $j$ we denote $T_j$ as the set of cells with brand $j$ and $T^j$ as the quadtree associated with the cells in $T_i$ for all $i \le j$:

\begin{definition}[
sets $T_j$]
\label{def:tree}
Let $T_1$ be a set of true cells in $\mathbb{R}^d$. 
We define the sets $T_j, 2 \leq j \leq d+1$ recursively:

Given a set of cells $T_j$ in $\mathbb{R}^d$, 
let  $T^j$ be  the quadtree given by $\cup_{i \le j }T_i$.

We define the set $T_{j+1}$ 
to be the minimal set of cells 
obtained by splitting cells of  $T^j$, such that 
each cell in $T_j$
is $2^j$-smooth in $T_{j+1}$.

For each set  $T_j$, to every cell in $T_j$ we assign brand $j$. 
\end{definition}

\begin{definition}
\label{def:t-star}
In $\mathbb{R}^d$, we define $T^*$ to be $T^{d+1}$.
\end{definition}

The extended quadtree $T^*$ has three useful properties which we 
prove in the remainder of this section: the tree is unique, the tree has a size linear in $d$ and cells in $T^*$ which are related in ancestry must have a related brand. 

\begin{lemma}
\label{lemma:unique}
Given a set of cells $T_j$ in $\mathbb{R}^d$ with brand $j$, the cells in $T_{j+1}$ that balance the cells in $T_j$ are unique.
\end{lemma}

\begin{proof}
Per definition a cell $C$ is $2^j$-smooth if all its neighboring leaf cells are at most a factor $2^j$ larger than $C$. This means that if we want to balance a cell $C$ then we need to check for each of its neighboring cells if it is too large and if so, add a minimum number of cells accordingly. This makes the minimum set of cells that balances a cell $C$ unique. If for each cell in $T_j$, its balancing cells are unique, then the set $T_{j+1}$ (the union of all sets of balancing cells) is unique.
\end{proof}

\begin{figure}[H]
	\centering
	\includegraphics[width=100px]{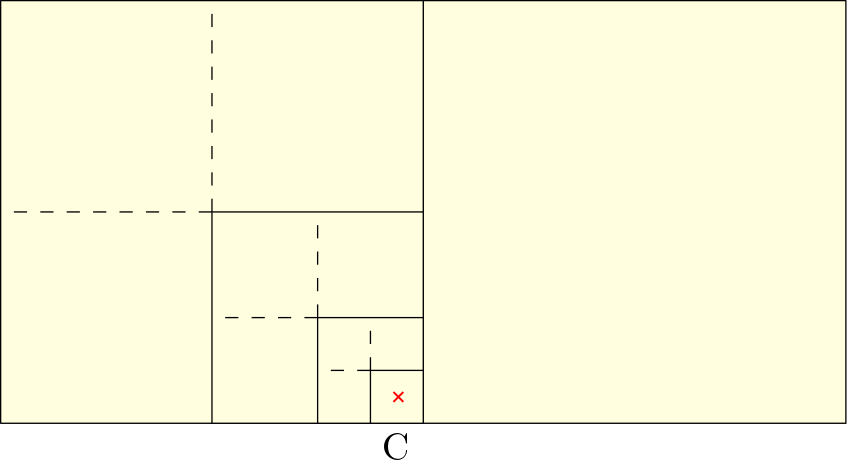}
	\caption{Let the figure show $T_{j-1} = T_{1}$, the true cells of a quadtree in $\mathbb{R}^2$ in white and denote the cell with the red cross as $C$. Cells shown with dotted lines exist but are not important for the example. Note that $C$ is adjacent to a cell of 8 times its size so $C$ is not 2-smooth and the parent of $C$ is also not 2-smooth. If we want to split the neighbor of $C$ into cells with brand 2 we create cells which are $4$ times the size of $C$ and we thus balance its parent. The second split creates cells of brand 2 which are twice the size of $C$ and so $C$ is 2-smooth. 
	\label{fig:splittinglemma}}
\end{figure}

\begin{lemma}
\label{lemma:treesize}
Every set of balancing cells $T_j$ has $\mathcal{O}(2^dn)$ cells. 
\end{lemma}

\begin{proof} 
We prove this by induction. By definition $T_1$ 
has $\mathcal{O}(n)$ cells, and all cells in $T_j$ exist to balance cells in $T_{j-1}$. Each of the $\mathcal{O}(n)$ cells $C$ in $T_{j-1}$ has at most $d$ leaf neighbors which are larger than $C$. If $C$ is not $2^{j-1}$-smooth we need to split the too large leaf neighbors to make $C$ $2^{j-1}$-smooth. Observe that for each split we either improve the balance of $C$ or of an ancestor of $C$ which is also in $T_{j-1}$ and which also had to be $2^{j-1}$-smooth (See Figure \ref{fig:splittinglemma}). The result is that we need at most $d$ splits to balance $C$.  
\end{proof}

\begin{corollary}
\label{cor:treesize}
If $d$ is constant, the 
tree $T^*$ has size $\mathcal{O}(n)$.
\end{corollary}

\begin{lemma}
\label{lemma:branding}
Let $C_1, C_2$ be two family related (possibly non-leaf) cells in $T^*$ such that $|C_1| \le |C_2|$. 
Then the brand of $C_2$ is at most the brand of $C_1$. 
\end{lemma} 

\begin{proof}
The proof is a proof per construction where we try to reconstruct the sequence of operations that led to the creation of cell $C_1$.  All of the ancestors of $C_1$ must have a brand lower or equal to the brand of $C_1$, this includes the parent 
$C_a$ of $C_2$. 
Since $C_1$ is a descendant of $C_a$, 
$C_a$ must be split. In that split all of the children of $C_a$ (including $C_2$) 
are created with a brand lower or equal to the brand of $C_1$. 
\end{proof}

\begin{lemma}[The Branding Principle]
\label{lemma:brandingprinciple}
Let $C_j$ be a cell in $T^*$ with brand $j$. Then all neighboring cells $N$ for which $|N| \ge 2^j|C_j|$ must have a brand of at most $j+1$.
\end{lemma}

\begin{proof}
This property follows from the definition of each set of cells $T_j$. If $C_j$ 
is $2^j$-smooth, its neighboring cells $N$ can be at most a factor $2^j$ larger than $C_j$. When we define $T_{j+1}$, all the neighbors of $C_j$ either already have 
 size 
at most $2^j|C_j|$ and thus a brand of at most $j$, or the neighbors must get split until they have 
size 
exactly $2^j|C_j|$. When the latter happens those cells get 
brand $j+1$. 
\end{proof}
With these lemmas in place we are ready to prove the main result for static uncompressed quadtrees in $\mathbb{R}^1$ and $\mathbb{R}^2$.

\subsection{Static uncompressed smooth quadtrees over $\mathbb{R}^1$.}

Let $T_1$ be a non-compressed quadtree over $\mathbb{R}^1$ which takes $\mathcal{O}(n)$ space. In this subsection we show that we can add at most $\mathcal{O}(n)$ cells to the quadtree $T_1$ such that all the cells in the resulting quadtree are $2^j$-smooth for some $j \le 2$ and the true cells are $2$-smooth. Lemma~\ref{lemma:treesize} tells us that we can add at most $\mathcal{O}(n)$ cells with brand $2$ to $T_1$ resulting in the tree $T^* = T_1 \cup T_2$ where all the true cells are  $2$-smooth in $T^*$. Our claim is that in a static non-compressed quadtree over $\mathbb{R}^1$ all the cells in $T_2$ must be $4$-smooth in $T^*$ since we cannot have two neighboring leaf cells in $T_2$ with one cell more than a factor 2 larger than the other.

\begin{theorem}
\label{theorem:R1}
Let $T_1$ be an uncompressed quadtree over $\mathbb{R}^1$ which takes $\mathcal{O}(n)$ space. In the smooth tree $T^*$ there cannot be two neighboring leaf cells $C_2$, $C_3$, both with brand $2$ such that $|C_2| \le \frac{1}{2^2}|C_3|$. \footnote{
Careful readers can observe two things in this section: (i) Cells which are 2-smooth are allowed to have neighbors which are 4 times as large but in $\mathbb{R}^1$ they cannot. (ii) The proof of this theorem actually shows that $C_3$ can not even be a factor two larger than $C_2$. 
We choose not to tighten the bounds because these two observations do not generalize to higher dimensions.}
\end{theorem}

\begin{figure}[H]
	\centering
	\includegraphics[width=300px]{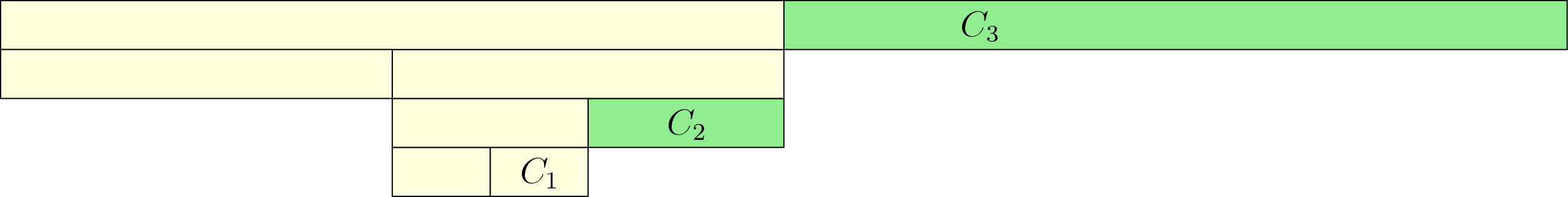}
	\caption{Two neighboring cells with brand $2$ in a one-dimensional quadtree.
    In the figure white cells have brand $1$ and light green cells have brand $2$.}
	\label{fig:onedimproof}
\end{figure}

\begin{proof} 
The proof is by contradiction, and is illustrated in  Figure \ref{fig:onedimproof}.
Assume for the sake of contradiction that we have two neighboring cells $C_2$ and $C_3$ both with brand 2 with $|C_3| = 4|C_2|$.  $C_2$ has two neighbors:
 one family related neighbor and one non-sibling neighbor. $C_3$ cannot be contained in a sibling neighbor because $C_3$ is larger than $C_2$. Note that $C_2$ exists to balance a true cell $C_1$ of smaller size. $C_1$ cannot be a descendant of $C_3$ because $C_3$ is a cell with brand 2. So $C_1$ must be a descendant of the sibling neighbor of $C_2$. That would make $C_2$ and $C_1$ family related and Lemma \ref{lemma:branding} then demands that $C_2$ has 
brand 
at most 1; 
a contradiction.
\end{proof}

\subsection{Static uncompressed smooth quadtrees over $\mathbb{R}^2$.}
\label{sec:R2}

We also show that we can make a smooth non-compressed static quadtree over $\mathbb{R}^2$ that takes $\mathcal{O}(n)$ space, such that all the cells in the quadtree are $2^j$-smooth for a $j \le 3$. We denote the original cells by $T_1$ and we want them to be $2$-smooth. We claim that in the extended quadtree $T^*$ (as defined in Definition \ref{def:t-star}) all cells are $8$-smooth.

\begin{theorem}
\label{theorem:R2}
Let $T_1$ be an uncompressed quadtree over $\mathbb{R}^2$ which takes $\mathcal{O}(n)$ space. In the extended tree $T^*$ there cannot be two neighboring leaf cells $C_3$, $C_4$ with both brand $3$ such that $|C_3| \le \frac{1}{2^3}|C_4|$.
\end{theorem}

\begin{proof}

The proof resembles the proof in Theorem \ref{theorem:R1}, and it is illustrated by Figure \ref{fig:twodimproof}. 
However, it requires two cases instead of one. Note that for $C_3$ to exist there must be at least two consecutive neighbors of $C_3$, ($C_2$ and $C_1$) with brand $2$ and $1$ respectively, such that $|C_1| = \frac{1}{2}|C_2| = \frac{1}{2} \cdot \frac{1}{4}|C_3|$. \\
\noindent \begin{minipage}{0.55\textwidth}
\-\hspace{15pt} Observe that $C_1$, $C_2$ and $C_3$ cannot be family related because of Lemma \ref{lemma:branding} and observe that $C_4$ can not be a sibling neighbor of $C_3$. The proof claims that it is impossible to place $C_1$, $C_2$, $C_3$ and $C_4$ in the plane without either violating the branding principle, Lemma \ref{lemma:branding} or causing a cell with brand $1$ or $2$ to be not smooth.

\-\hspace{15pt} Our first claim is that $C_3$ must share a vertex with $C_4$ (and similarly, $C_2$ must share a vertex with $C_3$).  If this is not the case all the neighbors of $C_3$ (apart from $C_4$) are either contained in sibling neighbors of $C_3$ or neighbors of $C_4$. However that would imply that either $C_2$ is \textbf{family related} to $C_3$ or that an ancestor of $C_2$ of size $|C_3|$ is a neighbor of $C_4$. The first case cannot happen because of Lemma \ref{lemma:branding}, in the second case we have a cell with brand 2 neighboring $C_4$ which is $\frac{1}{8}$'th the size of $C_4$ so $C_4$ must have been split but $C_4$ must be a leaf. Without loss of generality we say that $C_3$ shares the top left vertex with $C_4$ (Figure \ref{fig:twodimproof}).
\end{minipage}\hfill
\begin{minipage}{0.4\textwidth}
   \includegraphics[width=160px]{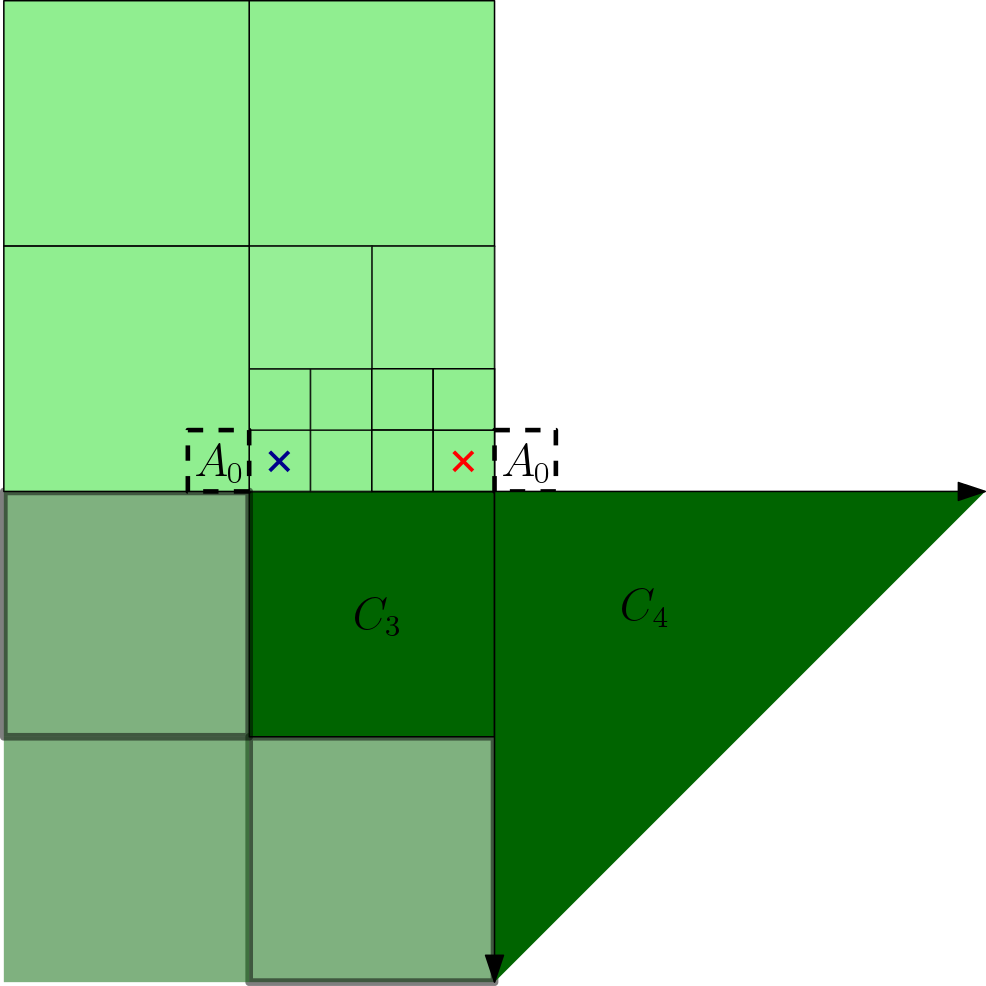}
	\captionof{figure}{Two neighboring cells with brand $3$ in a two-dimensional quadtree. The cells with brand 2 are light green and cells with brand 3 are dark green.}
	\label{fig:twodimproof}
\end{minipage}

Since $C_2$ cannot be placed in a sibling neighbor of $C_3$, $C_2$ must be placed in the positive $\vec{y}$ direction from $C_3$. $C_2$ must also share a vertex with $C_3$ so we distinguish between two cases: $C_2$ shares the top left vertex with $C_3$ or the top right.

\textbf{Case 1: The top left vertex.} In this case $C_2$ is the blue square in Figure \ref{fig:twodimproof}. $C_1$ cannot be contained in a sibling neighbor of $C_2$ so $C_1$ must lie to the left. However if $C_1$ is adjacent to $C_2$, its parent $A_0$ (the dashed lines in the figure) must also be a neighbor of $C_2$. Because we placed $C_4$ and $C_3$ without loss of generality, Figure \ref{fig:twodimproof} shows us that $A_0$ must neighbor a sibling neighbor of $C_3$ which we will denote as $F(C_3)$. We know that $|F(C_3)| = |C_3| \ge 2^1|A_0|$ and that $A_0$ has brand 1 and $F(C_3)$ has brand 3. This is a contradiction with the \emph{Branding Principle} (Lemma \ref{lemma:brandingprinciple}).

\textbf{Case 2: The top right vertex.} In this case $C_2$ is the red square in the figure. $C_1$ cannot be contained in a sibling neighbor of $C_2$ so $C_1$ must lie to the right. However, if $C_1$ is adjacent to $C_2$, its parent $A_0$ (the dashed lines in the figure) must also be a neighbor of $C_2$. Moreover (for similar reasons as the first case) $A_0$ must also be a neighbor of $C_4$. We know that $A_0$ is the ancestor of a true cell, so $A_0$ has brand 1. Moreover, $|C_4| \ge 8|A_0|$ so $C_4$ must have been split which contradicts that $C_4$ is a leaf.

Both cases lead to a contradiction so Theorem \ref{theorem:R2} is proven. The structure of this proof is equal to the structure of the proof of the generalized theorem in Section \ref{sec:dimensions}. \end{proof}

\section{Dynamic quadtrees.}
\label{sec:dynamic}

In the previous section we have shown that, given a static uncompressed quadtree $T_1$ of $\mathcal{O}(n)$ size over $\mathbb{R}^1$ or $\mathbb{R}^2$, we can create a static smooth tree $T^*$ of $\mathcal{O}(n)$ size.  In this section we prove that if $T_1$ is a dynamic tree, we can also dynamically maintain its extended variant $T^*$. 

Let $T_1$ be a quadtree over $\mathbb{R}^1$ or $\mathbb{R}^2$ subject to the split and merge operation (Table \ref{tab:operations} I.). If we use the split operator to create new true cells in $T_1$ then in $T^*$ we (possibly) need to add cells (to the set $T_2$) that smooth the new true cells. Similarly if we add cells to $T_2$ we might need to add cells to $T_3$. The first question that we ask is: can we create a new split and merge operator that takes a constant number of steps
per split and merge to maintain the extended quadtree $T^*$?\footnote{With $T^*$ as defined in Definition~\ref{def:t-star} and with $T_1 \subset T^*$.}

\begin{lemma}
\label{lemma:splitcount}
Given an uncompressed non-smooth quadtree $T_1$ in $\R^d$ of $\mathcal{O}(n)$ size and its extended tree $T^*$. Let $T_1'$ be an uncompressed non-smooth quadtree such that $T_1$ can become $T_1'$ with one merge or split operation, then $T^*$ and $T'^{*}$ differ in at most $(2d)^d$ quadtree cells. 
\end{lemma}

\begin{proof}
In the proof of Lemma \ref{lemma:treesize} we showed that each cell in $T_j$ had at most $d$ balancing cells in $T_{j+1}$. So if we add a new cell in $T_1$ with the split operation, we need to do at most $d^d$ split operations to create the $d^d$ cells to smooth the tree up to the level $(d+1)$ \footnote{Section \ref{sec:dimensions} will show that the $(d+1)$'th level is always $2^{d+1}$-smooth}. $d^d$ split operations create at most $(2d)^d$ cells. 

Lemma \ref{lemma:unique} states that for each set of true cells $T_1$ $T^*$ is unique. So if we want to \textbf{merge} four cells in $T_1$ we must get a new unique $T_1'$ and $T'^{*}$. We know that we can go from $T'^{*}$ to $T^*$ with $d^d$ \textbf{split operations}, so we can also go from $T^*$ to $T^{*}$ with $d^d$ \textbf{merge} operations.
\end{proof}

\subsection{The algorithm that maintains $T^*$.}

Lemma \ref{lemma:splitcount} tells us that it should be possible to dynamically maintain our extended quadtree $T^*$ with $\mathcal{O}((2d)^d)$ operations per split or merge in the true tree $T_1$. The lemma does not specify which cells exactly need to be split. We note that our extended quadtree $T^*$ is unique and thus independent from the order in which we split cells in $T_1$. In Appendix~\ref{app:naive} we show that this property prohibits the naive implementation (just split all cells which conflict with another cell's smoothness): this does not maintain a quadtree that follows the definition of $T^*$ given by Definition \ref{def:tree} (and thus does not have to be smooth). Instead we introduce the following lemma which will help us design a correct algorithm for maintaining $T^*$: 

\begin{lemma}
\label{lemma:equivalence}
Given a dynamic quadtree $T_1$ and its dynamic extended quadtree $T^*$ with $T_1 \subset T^*$ where $T^*$ is defined according to Definition \ref{def:tree}. If a cell $C \in T^*$ has brand $j$ then there is at least one neighbor $N$ of $C$ such that $N$ has brand $j-1$ and $|C| = 2^{j-1}|N|$.
\end{lemma}

\begin{proof}
If $C$ has brand $j$, then according to Definition \ref{def:tree}, $C$ exists to smooth a cell $N \in  T_{j-1}$. Per definition $N$ has brand $j-1$ so $|C|$ is indeed $2^{j-1}|N|$. 
\end{proof}

This observation allows us to devise an algorithm that maintains $T^*$ after a split operation in $T_1$, we call this algorithm the \textbf{aftersplit procedure}.
The first change to our static construction is that each cell $C \in T^*$ gets a collection of brands (which can contain duplicates). The current brand of a cell is the minimum of its collection of brands. Given this new definition of a brand we define a two-phased procedure: 
\begin{itemize}
\item 
Whenever we split a cell $C$ with a brand $j$, we check all the $d$ neighboring \emph{leaf} cells $N$ which are larger than $C$. If $N$ is more than a factor $2^{j-1}$ larger than $C$, the new children of $C$ are non-smooth and $N$ should be split into cells with brand $j+1$. If a neighbor $N$ is split, we also invoke the aftersplit procedure on $N$.
\item
Secondly, if for any neighboring leaf $N$ of $C$ there exists a cell $C'$ with $C'$ equal to $N$ or an ancestor of $N$ such that $C'$ is exactly a factor $2^{j-1}$ larger than $C$, $C'$ could exist to smooth the new children of $C$ so we add the brand $(j+1)$ to its set of brands. We call this \textbf{rebranding}.
\end{itemize}

\begin{algorithm}[h]
\caption{The procedure for after splitting a cell.}
\begin{algorithmic}[1]
\Procedure{AfterSplit(Cell $C$, Integer $j$)}{}
\For{Cell $N$ $\in$ LargerLeafNeighbors}
\If{$\frac{|N|}{|C|} > 2^{j-1}$}
\State $V \gets Split(N, j+1)$
\State $AfterSplit(N, j+1)$
\EndIf
\If{$\exists C' \in N\cup Ancestors(N)$ such that $\frac{|C'|}{|C|} = 2^{j-1}$}
\State $C'.Brands.Add(j+1)$
\If{Changed(C'.Brands.Minimum)}
\State $AfterSplit(C', C'.Brands.Minimum)$
\EndIf
\EndIf
\EndFor
\EndProcedure
\end{algorithmic}
\end{algorithm}

\begin{lemma}
If we split a cell with brand $j$, the aftersplit procedure performs at most $(2d)^{d-j}$ split and merge operations and the resulting extended quadtree $T^*$ implements Definition \ref{def:tree}. 
\end{lemma}

\begin{proof}
Observe that if we split a cell $C$ with brand $j$, $C$ has at most $d$ neighbors which are at least a factor $2^{j-1}$ larger than $C$ and a leaf in $T^*$. We can find the larger leaf neighbors with at most $d$ level pointer traversals and the neighbors of exactly $2^{j-1}$ size by first finding an ancestor of $C$ and using that ancestor's level pointers. If we \textbf{rebrand} or split one of the found neighbors, that neighbor gets a brand one greater than $j$ so the aftersplit procedure will recurse with a new $j' = j+1$. This means that we recurse at most $d-j$ times which makes the aftersplit procedure on a cell with brand $j$ perform $d^{d-j}$ split operations which creates at most $(2d)^{d-j}$ cells.

The resulting tree $T^*$ must implement Definition \ref{def:tree} because of Lemma \ref{lemma:equivalence} which shows that any neighboring cell $N$ with $|N| = 2^{j-1}|C|$ could therefore exist to smooth a child cell of $C$ in the static scenario. 
\end{proof}

The AfterMerge procedure is simply the inverse of the Aftersplit procedure. If we merge cells into a cell $C$, we check all neighbors of size $2^{j-1}|C|$ and $2^{j-2}|C|$. In the first case, we remove the brand $j$ from the cell and check if it still has to exist. In the second case, we remove the brand $j$ from the cell and check if it needs to be rebranded to a higher brand. 

\begin{theorem}
\label{theorem:dynamic}
For each dynamic compressed quadtree $T_1$ over $\mathbb{R}^d$ we can maintain a extended variant $T^*$ with at most $\mathcal{O}((2d)^d)$ operations per split or merge on $T_1$.
\end{theorem}

\section{Quadtrees in higher dimensions.}
\label{sec:dimensions}

In this section we prove that we can dynamically maintain smooth uncompressed quadtrees in $\mathbb{R}^d$. Assume that we have an uncompressed quadtree $T_1$ over $\mathbb{R}^d$. We dynamically maintain $T^*$ with the operations specified in Section \ref{sec:dynamic} in $(2d)^d$ time split or merge operation per operation on $T_1$. We claim that the resulting tree $T^*$ all uncompressed components are smooth. This claim is a direct result from the following theorem:

\begin{theorem}
\label{theorem:RD}
Let $T_1$ be  an uncompressed quadtree over $\mathbb{R}^d$ which takes $\mathcal{O}(n)$ space. In the extended tree $T^*$ there cannot be two neighboring leaf cells which we will name $C_{d+1}$, $C_{d+2}$ with both brand $(d+1)$ such that $|C_{d+1}| \le \frac{1}{2^{d+1}}|C_{d+2}|$.
\end{theorem}

\subsection{Necessary operations and lemmas.}

To prove Theorem \ref{theorem:RD} we need a few concepts. In the previous proof for the version for $\mathbb{R}^2$ we noted that each cell with brand $3$ needed a (possibly non-leaf) neighbor with brand $2$ that forced its existence. Similarly each cell of brand $2$ needs a neighbor with brand $1$ that forces its existence. Intuitively this creates a \textbf{chain} of cells, a concept which we can formalize:

\begin{definition}[Chain]
In our extended quadtree $T^*$, we call an ordered set $C = \{C_i \in T^* \mid i \in [a,b]\}$ a \textbf{chain} of length $(b-a) = k$ if: for each $j$, $C_j$ has brand $j$ and if for each $j \in [a,b-1]$ holds: $C_j$ neighbors $C_{j+1}$ and $|C_{j+1}| = 2^j| C_{j}|$. 
\end{definition}

\begin{lemma}
\label{lemma:relation}
Given a chain of length $k$, for any $j \le k-1$, the cells in the chain $C_j$ and $C_{j+1}$ cannot be family related.
\end{lemma}

\begin{proof}
Per definition, $C_{j}$ has a smaller brand than $C_{j+1}$ and $|C_j| < |C_{j+1}|$, but the branding principle says that if $C_j$ and $C_{j+1}$ are family related, $C_{j+1}$ can have a brand of at most $j$ which is a contradiction.
\end{proof}

We intuitively prove that we can maintain our extended uncompressed quadtree $T^*$ through showing that no chain of size $d+1$ can exist. To prove this we define a virtual operation which we will call the \textbf{step operator}. This operator defines how one can traverse our uncompressed quadtree. Given a quadtree cell and an integer $j$, we can use the step operator to find a cell in each of the $2d$ \emph{cardinal directions}.

\begin{definition}[STEP operator]
Given a constant $j$, a cell $A$ which is $2^i$-smooth with $i \le j$ and a cardinal direction in $\mathbb{R}^d$ $\pm \vec{v}$, we define the STEP operator as follows: $STEP(A, j, \pm \vec{v})$, finds the unique cell in our quadtree $T^*$ of size $2^j|C|$ that is an ancestor of the unique virtual level neighbor of $A$ which shares the facet of $A$ in the direction $\pm \vec{v}$. This cell can be an ancestor or a neighbor of $A$ and does not have to be a leaf. See figure \ref{fig:stepexample} for an example.

Given an ordered sequence of vectors $\vec{D}$; $STEP(A, j, \vec{D})$ iteratively applies the step operator with the vectors of $\vec{D}$ and each time raises $j$ by one. If $|\vec{D}| = k$ we denote the result as $STEP_k(A)$. 
\end{definition}

\begin{figure}[h]
	\centering
	\includegraphics[width=300px]{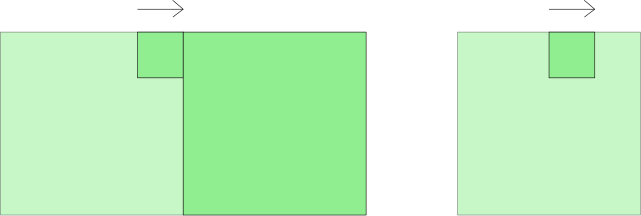}
	\caption{The step operator on a vector $+\vec{x}$, one result giving an ancestor and the other a neighboring cell.}
	\label{fig:stepexample}
\end{figure}

\begin{lemma}
\label{lemma:steps}
Given a cell $A$ which has brand $j$, we can apply the STEP operation on $A$ with $j$ as the step constant and any cardinal direction $\pm \vec{v}$ to always reach an existing cell in $T^*$. Given a chain from $C_i$ to $C_j$ with $i,j \in [d+1]$,  there must exist a collection of vectors such that we can traverse the chain using the STEP operator.
\end{lemma}

\begin{proof}
Each of the neighboring leaf cells of $A$ are per definition at most a factor $2^{j}$ larger than $A$. So for any cardinal direction $\pm \vec{v}$ we can find a neighboring leaf of $C$ which must have an ancestor cell which is exactly $2^j$ times larger than $A$.

In a chain the successive cells are exactly a factor $2^j$ larger and neighbors so for any chain from $C_i$ to $C_j$ we must be able to find a sequence of vectors such that the chain can be traversed using the STEP operator.
\end{proof}

\begin{definition}
Given an ordered sequence of vectors $\vec{D}$ used to traverse a chain of length $k$. We denote $\vec{D}[i]$ as the vector used to go from $C_{i-1}$ to $C_i$ and $\vec{D}[i, j]$ to denote the path used to go from $C_{i-1}$ to $C_j$.
\end{definition}

The second property that we need is intuitively that the neighbor relation is preserved by the step operator. The following lemma proves that if we have two cells which are neighbors in a dimension $\vec{y}$ and we traverse a path from both cells that the resulting two cells are either neighbors in $\vec{y}$ or the same cell.

\begin{lemma}
\label{lemma:equalsteps}
Let $C_1$ and $C_2$ be two neighboring cells of equal size. neighboring in a dimension $\vec{y}$. Let $\vec{D}$ be an ordered collection of cardinal directions, none of which are in dimension $\vec{y}$.

Then $STEP(C_1, \max\{C_1.Brand, C_2.Brand\}, \vec{D})$ and  $STEP(C_2, \max\{C_1.Brand, C_2.Brand\}, \vec{D})$\footnote{We take the maximum of both brands so that the algorithm always returns an existing cell. See Lemma \ref{lemma:steps}.} either return the same cell, or two cells which must be neighboring in $\vec{y}$.
\end{lemma}

\begin{figure}[h]
	\centering
	\includegraphics[width=150px]{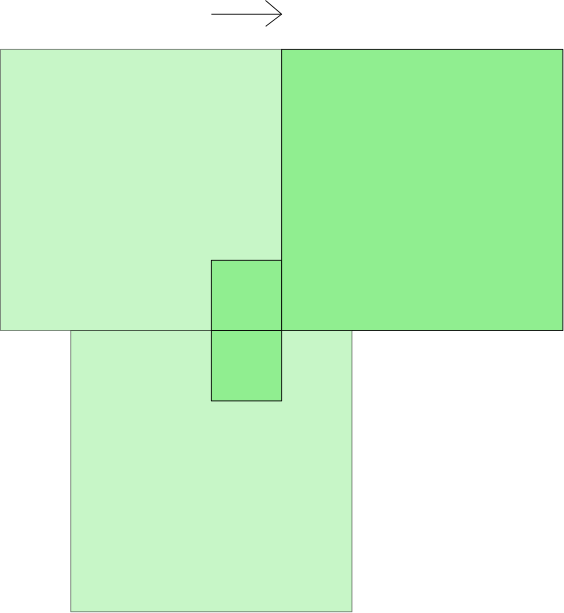}
	\caption{The step in the $\vec{x}$ direction, both resulting cells are clearly not aligned.}
	\label{fig:stepcounterexample}
\end{figure}

\begin{proof}
Assume the result cells $STEP_{|\vec{D}|}(C_1)$ and $STEP_{|\vec{D}|}(C_2)$ are not the same cell \underline{and} not neighboring in $\vec{y}$. Then there exists one lowest index $i$ for the step operator with the cardinal direction $\pm \vec{x} = \vec{D}[i]$ for which $STEP_{i}(C_1)$ is an ancestor of $STEP_{i-1}(C_1)$ whilst $STEP_{i}(C_2)$ is a neighbor of $STEP_{i-1}(C_2)$. However $C_1$ and $C_2$ are equally large and they start with the same step factor $j$. Which means that in all iterations the cells are equally large. If for one $i$ the STEP operator returns an ancestor cell in one iteration and an neighboring cell in the other, the quadtree grid must not be aligned. This argument is illustrated by Figure \ref{fig:stepexample}.
\end{proof}

\subsection{The proof.}

These properties, together with Lemma  \ref{lemma:unique}, \ref{lemma:branding} and \ref{lemma:brandingprinciple} from Section \ref{sec:static} allow us to prove Theorem \ref{theorem:RD}. In this subsection we prove Theorem \ref{theorem:RD} by proving the lemma below. The proof follows an identical structure to the proof of Theorem \ref{theorem:R2} in Section \ref{sec:static}.

\begin{lemma}
\label{lemma:nocascade}
The cell $C_{d+1}$ from Theorem \ref{theorem:RD} can only exist because of a chain of $d$ cells which we will denote $C_1 ... C_d$. We can only embed $C_{d+2}$, $C_{d+1}$ and the chain such that we can go from $C_1$ to $C_{d+2}$ using the step operation with a path $D$ with contains exactly one arrow per dimension.
\end{lemma}

\begin{corollary}
Given Lemma \ref{lemma:nocascade}, we know that we can never have $C_{d+1}$ neighboring to $C_{d+2}$ because we must place them in a path that uses $d+1$ dimensions, whilst in $\mathbb{R}^d$ we only have $d$ dimensions. This proves Theorem \ref{theorem:RD}.
\end{corollary}

The proof of Lemma \ref{lemma:nocascade} is a proof by contradiction. Given $C_{d+2}$ we want to place $C_{d+1}$ to $C_1$ in the plane in decreasing order. This placement implies an ordered set of cardinal vectors $\vec{D}$. Let $j+1$ be the highest index for which $\vec{D}[j+1]$ is a cardinal vector in a dimension $\vec{y}$ which is also used at a higher index. Let $k$ the higher index where we used of $\vec{y}$ the first time. In our earlier defined notation, we must have:

\begin{align*}
C_k &= STEP(C_{k-1}, \textit{ } k-1,\textit{ } D[k] = \pm \vec{y}) \\
C_{j+1} &= STEP(C_{j},\textit{ } j,\textit{ } D[j+1]  \pm \vec{y})
\end{align*}

We then distinguish between two cases: $C_{j+1}$ lies with respect to $C_j$ in the same dimension \underline{and} direction as $C_k$ with respect to  $C_{k-1}$  \textbf{or} in the same dimension \underline{and opposite direction}. These two cases correspond with the two cases in the proof of Theorem \ref{theorem:R2} in Section \ref{sec:R2}.

\begin{figure}[h]
	\centering
	\includegraphics[width=300px]{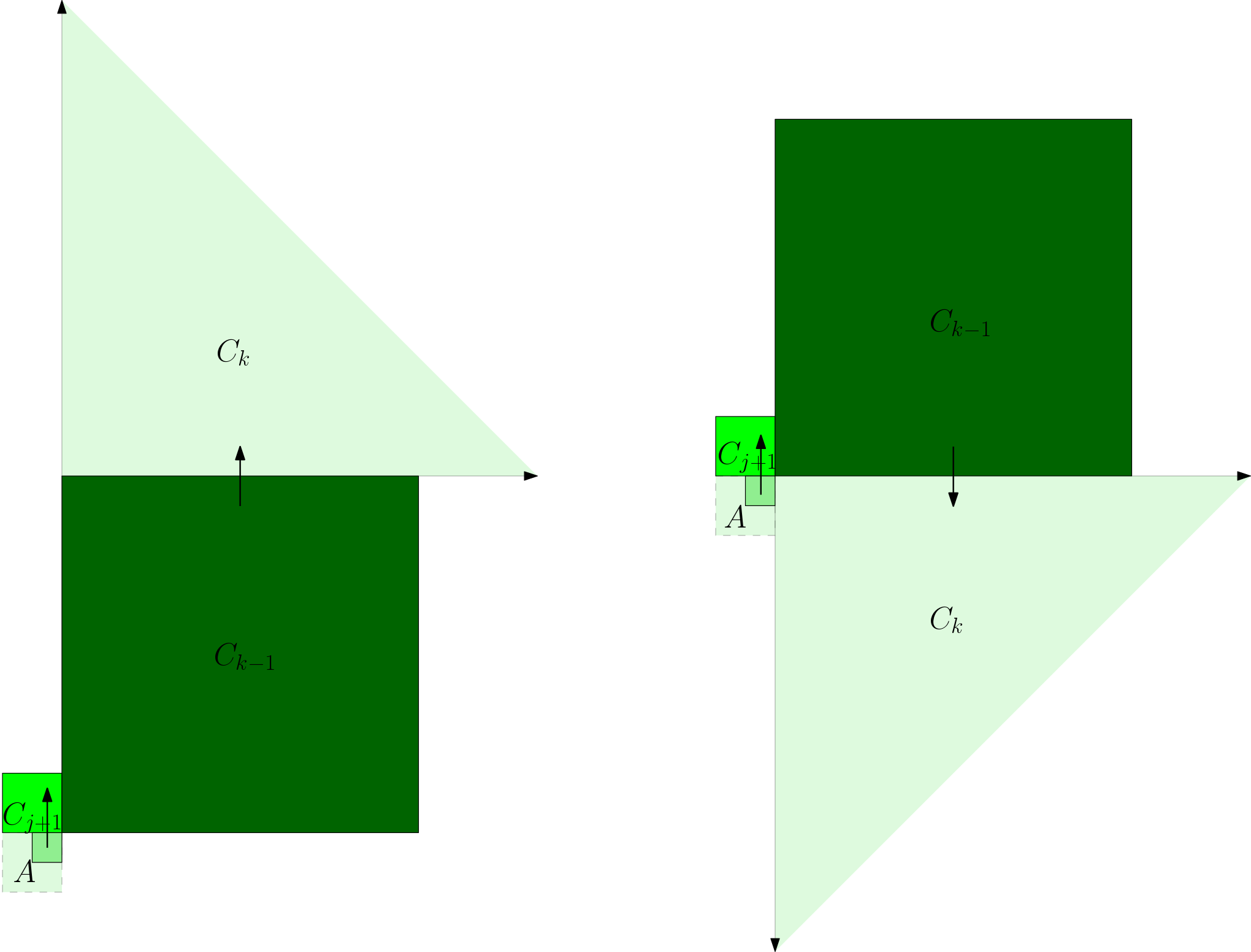}
	\caption{The two cases for the proof of Lemma \ref{lemma:nocascade}. On the left we see the cell $C_j$ in light green, $C_{j+1}$ in the $+\vec{y}$ direction in green, $C_{k-1}$ in the $+\vec{x}$ direction in dark green and $C_k$ in the $+\vec{y}$ direction. On the right we see the second case which is similar, but $C_k$ is in the $-\vec{y}$ direction. In both cases the cell $A$ of the proof is shown with dotted lines.}
	\label{fig:finalproof}
\end{figure}

\subsubsection{Case 1: The same direction.}

Let 
\begin{align*}
C_{j+1} &= STEP(C_{j}, j, + \vec{y}) \\
C_k &= STEP(C_{k-1}, k-1, +\vec{y}) \textit{ and W.l.o.g: }\\
C_{k-1} &= STEP(C_{k-2}, k-2, +\vec{x}) \\
\end{align*}

We can then draw $C_j, C_{j+1}$ and $C_k$ in the $(\vec{x}, \vec{y})$ plane (see Figure \ref{fig:finalproof}). Let $A$ be the ancestor of $C_j$ with $|A| = |C_{j+1}| = 2^j |C_j|$. Observe that the branding principle states that $A$ can be at most $2^j$-smooth.

\begin{lemma}
We know that there is a path from $C_j$ to $C_k$ that takes $k-j$ STEPS: That path consists out of $\vec{D}[j+1,k] = \vec{y} \rightarrow \vec{D}[j+2, k-2]  \rightarrow \vec{x} \rightarrow \vec{y}$ (see Figure \ref{fig:finalproof}). Our claim is that we can use $k-j-2$ STEPS from $A$ to find a cell $R$ which is \textbf{family related} to $C_{k-1}$.
\end{lemma}

\begin{proof}
The path with $k-j-2$ steps is given by  $\vec{D}[j+2, k-2]$. Observe that $A$ and $C_{j+1}$ are equally sized and neighboring in $\vec{y}$ and that $\vec{y} \notin \vec{D}[j+2,k-1]$. If we perform the STEP operator on $A$ and $C_{j+1}$ with the path $\vec{D}[j+2,k-1]$ Lemma \ref{lemma:equalsteps} states that the result cells are either the same cell or two cells neighboring in $\vec{y}$. The result cells cannot be the same cell: we know that for $C_j$ the result cell is $C_{k-1}$ and that cell must be $2^{k-1}$-smooth. Noting that the result cell $R$ from the path from $A$ can be at most $2^{j + (k-j-2)} = 2^{k-2}$-smooth gives a contradiction.
\end{proof}

So from $A$ we get a cell $R$ which neighbors $C_{k-1}$ in the $\vec{y}$ dimension. We know that after two steps $\vec{y} \rightarrow \vec{y}$, the result cell $R'$ must contain $C_k$. Now observe that if you take two consecutive STEPS in the same direction, one of the two steps must yield an ancestor. This means that either $R$ and $C_{k-1}$ are family related,\ or $C_{k-1}$ and $C_k$ are family related. But $R$, $C_{k-1}$ and $C_k$ have brand $(k-2)$, $(k-1)$ and $k$ respectively so Lemma  \ref{lemma:relation} tells us they cannot be family related.

\subsubsection{Case 2: The opposite direction.}

Let $A$ be the ancestor of $C_j$ such that $|A| = |C_{j+1}| = 2^j |C_j|$. Observe that the branding principle states that $A$ can be at most $2^j$-smooth. Moreover let the cardinal directions that use the $\vec{y}$ dimension twice point in the opposite direction, formally: $\vec{D}[j+1, k] = \vec{y} \rightarrow \vec{D}[j+2, k-2]  \rightarrow \vec{x} \rightarrow -\vec{y}$.

\begin{lemma}
We can use $k-j-3$ STEPS from $A$ to find a cell $R$ which is \textbf{neighboring} to $C_{k}$.
\end{lemma}

\begin{proof}
The path is given by $\vec{D}[j+2, k-2]$.  Observe that $A$ and $C_{j+1}$ are equally sized and neighboring in $\vec{y}$. If we perform the STEP operator on both with the path of size $(k-j-2)$: $\vec{D}[j+2, k-1]$, Lemma \ref{lemma:equalsteps} tells us that the result cells are either the same cell, or two cells neighboring in $\vec{y}$. The result cells cannot be the same cell: We know that for $C_j$, the result cell is $C_{k-1}$ and that cell must be $2^{k-1}$-smooth. However, the result cell $R$ from the path from $A$ can be at most $2^{j + (k-j-2)} = 2^{k-2}$-smooth.

So the path $\vec{D}[j+2,k-1]$ from $A$ yields a cell $R'$ which is a bottom neighbor of $C_{k-1}$ But we know that $C_k$ is below $C_{k-1}$ and neighboring to $C_{k-1}$ so our result cell $R$ must overlap with $C_{k-1}$.

We know because of Lemma \ref{lemma:relation}, that none of the STEPS in $\vec{D}$ can return an ancestor (because we have a chain, we must have an increasing brand but the branding principle then prohibits us finding an ancestor). This means that if we traverse only $\vec{D}[j+2, k-2]$ from $A$ without the last vector $+\vec{x}$, we find a cell $R$ neighboring to $C_k$.
\end{proof}

However, now we find a contradiction that proves Lemma \ref{lemma:nocascade} since $C_k$ is too large to be neighboring $R$: 

\begin{align*}
&|R| = |STEP(A, j, \vec{D}[j+3, k-2])| \le  \left( \prod_{i=j}^{< j + (k-j-3) } 2^i \right) \cdot |A| = \left( \prod_{i=j}^{< k-3} 2^i \right) \cdot 2^j \cdot |C_j| \\
&|C_k| = \left( \prod_{i=j}^{< j+ (k-j)} 2^i \right) \cdot |C_j| = 2^{k-1} \cdot 2^{k-2} \cdot 2^{k-3} \cdot  \left( \prod_{i=j}^{<k-3} 2^i \right) \cdot |C_j| \\
&\textit{Observing that R is } 2^{k-3}-smooth \textit{ implies that } R \textit{ can be next to a cell of maximally } \\
&2^{k-3} \cdot |R| = 2^{k-3} \cdot 2^j \cdot \left( \prod_{i=j}^{k-3} 2^i \right) \cdot |C_j| < |C_k| \qed
\end{align*}

\section{Compression.}
\label{sec:compression}

In this section and the next, our goal is to build smooth quadtrees that store a set of points and are dynamic with respect to insertions and deletions of points from the set.
The first new challenge this presents is that
we want the quadtree to have a size linear in the number of points and so we need compression. 
Compression subdivides the $\mathcal{O}(n)$ true cells $T_1$ into a set of \textbf{uncompressed components} $\mathcal{A}_1$ and simple paths linking uncompressed components.

Static smooth compressed quadtrees are a well-studied topic and they are defined in [link]. However, just as in Section \ref{sec:static} we do not make use of the standard definition of smooth quadtrees but define our own extended quadtree $T^*$ so that we can make the tree dynamic. In this section we first define what we mean with $2^j$-smooth cells in a compressed quadtree and then we formally define our extended quadtree $T^*$. Afterwards we show that we can dynamically maintain our extended quadtree with time depending only on the dimension $d$ per operation in the true compressed tree $T_1$.

In a compressed quadtree $T_1$, we only store the uncompressed components $\mathcal{A}_1$ and pointers for the compressed links. Our extended quadtree $T^*$ also consists out of a set of uncompressed components denoted by $\mathcal{A}^*$ and compressed links. Since cells in compressed links are not stored in memory, we only want to balance cells in uncompressed components.
\begin{definition}
Given any quadtree $T$, a cell $C$ in an uncompressed component of $T$ is $2^j$-smooth if for all (possibly compressed) neighboring cells $N$ in $T$ holds $|N| \le 2^j|C|$.
\end{definition}
\begin{observation}
\label{observation:unbalanced}
An a cell $C$ in an uncompressed component $A$ his brand can be violated with respect to a compressed or uncompressed cell in either its uncompressed component $A$ or another compressed component. The latter case can only happen if $C$ intersects the border of the root of $A$ (see Figure \ref{fig:compressionbalance}).
\end{observation}
\begin{figure}[H]
	\centering
	\includegraphics[width=300px]{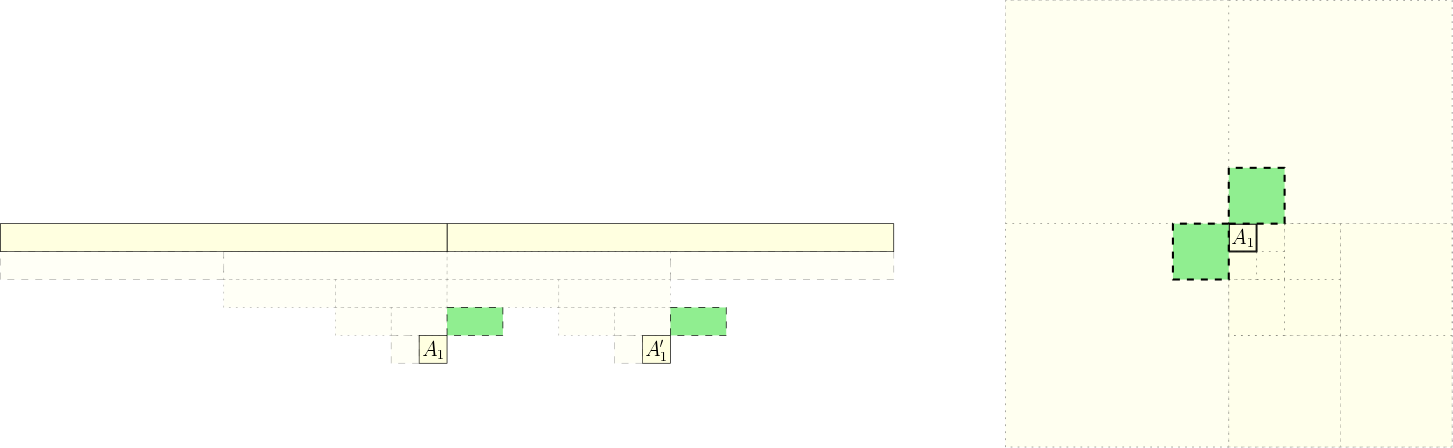}
	\caption{(left) Two roots of uncompressed components in $\mathbb{R}^1$, $A_1, A_1' \in \mathcal{A}$ in white with brand 1. Observe that their left neighbor is a sibling neighbor and that their right neighbor (though compressed) is more than a factor 2 larger than the root. To balance the roots we would need to insert the green cells with brand 2 (later we define these cells as supporting roots).  (right) A similar situation in $\mathbb{R}^2$.}
	\label{fig:compressionbalance}
\end{figure}

Recall that we defined our extended quadtree $T^*$ in Section \ref{sec:static} with Definition \ref{def:tree}. We had $(d+1)$ different sets of cells where each set $T_j$ was defined as the minimal number of cells needed to balance the cells in $T_{j-1}$. We showed that if we have two uncompressed quadtrees $T_1$ and $T_1'$ that differ in only one split/merge operations then their extended quadtrees $T^*$ and $T'^*$ only differ in at most $(2d)^d$ cells. However, if $T_1$ and $T_1'$ are uncompressed quadtrees and if we use Definition \ref{def:tree}, the reader can imagine that there are scenarios where their extended quadtrees $T^*$ and $T'^*$ differ in an arbitrarily large number of cells (see Section \ref{sec:compression}). This is why instead we redefine our extended quadtree as the union of all extended quadtrees $A^*$ of all uncompressed components $A_1$ of $T$:

\begin{definition}
\label{definition:T-star-final}
For each uncompressed component $A_1$ of $T_1$ we maintain its own extended quadtree $A^*$ according to definition \ref{def:tree} and we define the smooth extended quadtree of $T_1$ as $T^* := \cup_{A_1 \in \mathcal{A}_1} A^*$.\footnote{$A^*$ contains what we call the supporting roots of the root of $A_1$ see Section \ref{sec:supporting}.} 
\end{definition}

Note that we defined that cells only can violate their brand if they are cells in an uncompressed component and that thus because of Theorem \ref{theorem:RD} in Section \ref{sec:dimensions} the extended quadtree $T^*$ is a smooth quadtree. In the remainder of this section we will show how we can dynamically maintain $T^*$ with at most $\log(\alpha) \cdot d^2 \cdot (6d)^{d}$ operations per operation in the true tree $T_1$.

\subsection {Supporting roots.}
\label{sec:supporting}

Observation \ref{observation:unbalanced} showed that a cell $C$ in an uncompressed component $A$ can violate its brand because of a neighbor in either $A$ or another uncompressed component.  In Section~\ref{sec:applications}
 we showed that the latter case of unbalance must be resolved if we want level pointers to exist. If we want to balance a cell $C$ with respect to cells not in its compressed component we might need several compressed roots in our extended quadtree $T^*$ that contain cells that smooth the quadtree.

\begin{definition}
Let $R$ be a root cell with brand $j$ of an uncompressed component $A$ in our extended quadtree $T^*$. Let $R'$ be the root cell of an uncompressed component $A'$, such $|R'|>|R|$.

If $R'$ contains a cell $C'$ that exists to balance a cell contained in $A$ \emph{or} if $R'$ contains a cell $C'$ that exists to balance a cell that balances a cell in a supporting root of $R$ then we say that $R'$ is a \textbf{supporting root} of $R$ (see Figure \ref{fig:supportingroot} for an example).
\end{definition}

\begin{figure}[H]
	\centering
	\includegraphics[width=150px]{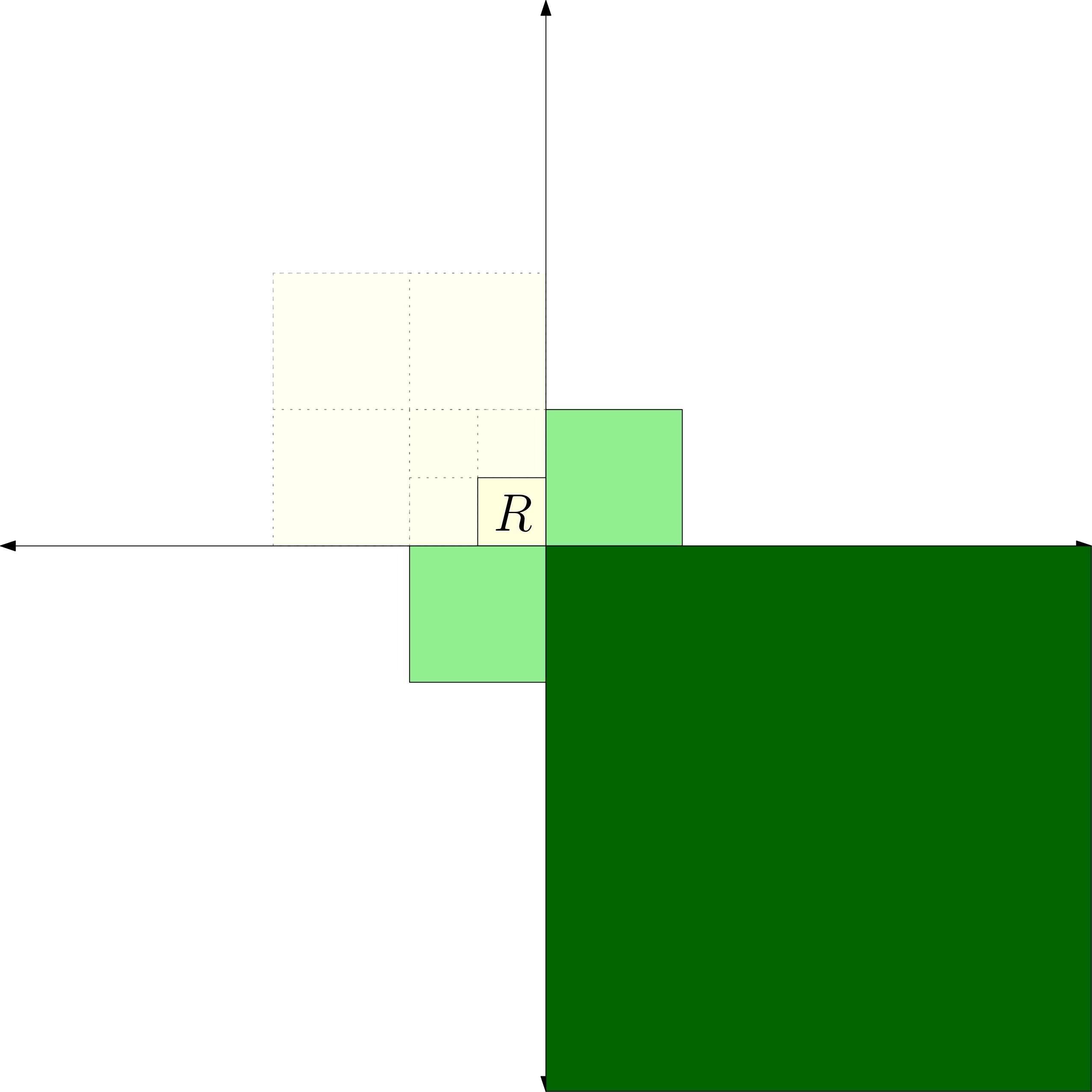}
	\caption{A root of a true cell $R$ and its three supporting roots. The light green cells have brand 2 and exist to balance $R$ itself. The dark green cell has brand 2 and exists to balance the other supporting roots of $R$. Observe that the sibling neighbors of $R$ could also be its supporting roots.}
	\label{fig:supportingroot}
\end{figure}

\begin{observation}
Let $R$ be a root cell with brand $j$ of an uncompressed component $A$ in our extended quadtree $T^*$. All the cells that balance descendants of $R$ are either contained in $R$ or contained in a supporting root of $R$.
\end{observation}

\begin{proof}
Let $C$ be a descendant of $R$ with a brand $i \ge 1$ such that $C$ violates its brand. Then the cell $C'$ with brand $i+1$ which we need to balance $C$ is either contained in $R$ or not. If $C$ is non-smooth with respect to a cell that is not contained in $R$, all the ancestors of $C$ also violate their brand and $C'$ needs an ancestor of at least size $|R|$ which is then a supporting root.
\end{proof}

\begin{lemma}
\label{lemma:consecutive}
 Let $A$, $B$ and $C$ be three equally sized cells with $A$ a neighbor of $B$ in any direction $+\vec{x}$ and $B$ of $C$ in $+\vec{x}$. Then if we consecutively split $A$ we might need to split $B$ for balance but never $C$. 
\end{lemma}

\begin{proof}
If we split $B$ to maintain balance with $A$, only half of its children neighbor $A$ and those children cannot neighbor $C$. So only cells which do not neighbor $C$ can be further split to maintain balance with $A$ which means that $C$ will never be split to maintain balance with cells in $B$.
\end{proof}

\begin{corollary}
\label{corollary:perimitersize}
Let $R$ be the root of an uncompressed component $A_1$ in the true tree $T_1$. Then there can be at most $3^d$ supporting roots of $R$. 
\end{corollary}

\begin{proof}
$R$ has at most $3^d$ cells in its perimeter with size $|R|$. Note that all cells that balance cells in $A_1$ must have a supporting root that neighbors $R$. It now follows from Lemma \ref{lemma:consecutive} that all the cells that balance the cells that balance cells in $R$ must be contained in the perimeter of $R$ and we can thus have at most $3^d$ supporting roots of $R$.
\end{proof}

In Section \ref{sec:preliminaries} we explained that a compressed quadtree was subject to five operations: \textbf{merge and split} within an uncompressed component, \textbf{insertion and deletion} of a compressed root and \textbf{upgrowing}. In Section \ref{sec:dynamic} we showed that after the first two operations on $T_1$ we can maintain $T^*$ in $\mathcal{O}((2d)^d)$ time per operation. In this section we show that we can maintain $T^*$ in $\mathcal{O}(3^d(2d)^d)$ time per operation on $T_1$. With Corollary \ref{corollary:perimitersize} we show in the following sections that we can maintain our extended quadtree $T^*$ with only five operations: merge and split operations in uncompressed components, insertions and deletions of a compressed root and \textbf{rebranding}. 

\subsection{Dynamically maintaining $T^*$.}

Assume we want to add a point $p$ to our point set $P$. We explained in Section \ref{sec:preliminaries} that we are given a pointer to the leaf cell $C$ that contains $p$. We dynamically maintain $T^*$ and update $T_1$ and $T^*$ by differentiating between two cases:
$C$ is either a leaf of an uncompressed component or $C$ is the parent of a compressed root $R$ and $p$ lies in $C \setminus R$ (The \textbf{doughnut} of $C$):

\subsubsection{Case 1: $C$ is a leaf.}
If $C$ is a \textbf{leaf} of an uncompressed component then either it does not contain a point yet, in which case we just store the new point there, or it does contain a point and we need to split the leaf, which leads to possibly a new compressed link. In the first case, we can restore $T^*$ in $\mathcal{O}((2d)^d)$ time. In the second case we add a new root $R$ to $T_1$ in constant time. In $T^*$ we would also need to add the at most $3^d$ supporting roots in its perimeter. 

Since all uncompressed components take $\mathcal{O}((2d)^dn)$ space in total, $T^*$ will take at most $\mathcal{O}(d(2d)^dn)$ space (Lemma \ref{lemma:treesize}). The only problem with this definition is that Lemma \ref{lemma:branding} (which said that family related cells have a related brand) does not have to hold anymore: in our extended quadtree $T^*$, we can have cells which have child cells with a lower brand. That lemma is crucial for the proof of Theorem \ref{theorem:RD} which is why we add the following operations:
\begin{itemize}
\item The \textbf{afterInsertCompressed} procedure is invoked right after we insert a root of a compressed component with brand $j$. The procedure first checks if the compressed root has actual stored ancestors within $d^2$ depth from the root. These ancestors always must be balancing cells with a higher brand than the root else the root did not have to be inserted with a compressed link. We rebrand all actual stored ancestors within a depth $d^2$ from the root with the brand $j$. The procedure will not add any cells of brand $j$ if the ancestors are compressed and takes at most $\mathcal{O}(d^2 \cdot (2d)^d)$ time per inserted root.
\item The \textbf{afterDeleteCompressed} procedure is invoked right after we delete a root of a compressed component. The procedure checks at most $d^2$ actual stored ancestors of the deleted root and checks whether or not they can return to their original brand. This also takes $\mathcal{O}(d^2 (2d)^d)$ time per deleted root.
\end{itemize}
This together with Corollary \ref{corollary:perimitersize} guarantees that we can insert a root $R$ of a new compressed component in a leaf in $T_1$ in constant time and update our extended quadtree $T^*$ with $\mathcal{O}(3^d d^2(2d)^d)$ operations such that $T^*$ is $2^{d+1}$-smooth.

\subsubsection{Case 2: $C$ is a doughnut.}
\label{sec:decomp}

\begin{figure}[H]
	\centering
	\includegraphics[width=300px]{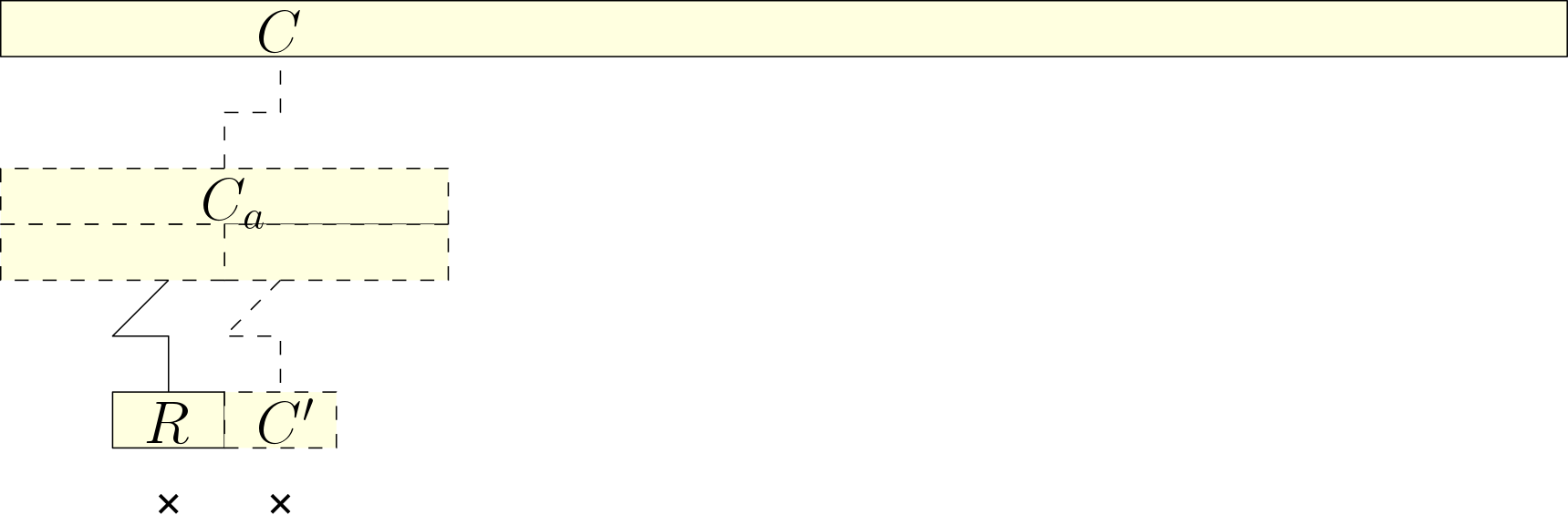}
	\caption{A common ancestor $C_a$ with a compressed link to the original cell $C$ which has two child cells with each a compressed link to the original root $R$ and the new leaf $C'$.}
	\label{fig:compressedlink}
\end{figure}

Let the point $p$ be contained in the doughnut which has a corresponding root of an uncompressed component $R$. Let $C'$ be the leaf cell that we want to insert in the quadtree and $C_a$ be the first common ancestor of $R$ and $C'$ (See Figure \ref{fig:compressedlink}. If $C$ is a doughnut, we make three case distinctions:

If $C_a$ is more than $\log(\alpha)$ levels of depth away from $R$, $C'$ and $Parent(C)$, we get three compressed links instead of only one, each with their own supporting roots and we do $3 \cdot d^2 \cdot (6d)^d$ work.

If $C_a$ is less than $\log(\alpha)$ levels of depth away from $R$, $C'$ and $Parent(C)$, we need to \textbf{decompress} the compressed link and possibly the compressed links of the supporting roots of $R$ in $\mathcal{O}(\alpha 6^d)$ time.

If $C_a$ is less than $\log(\alpha)$ levels of depth away from $R$ and $C'$ but more than $\log(\alpha)$ levels of depth away from $Parent(C)$ then $C_a$ needs to be connected to both $R$ and $C'$ and $C_a$ becomes the new root of the uncompressed component. In this last case we perform the \textbf{upgrowing} operation at most $\log(\alpha)$ times. A single upgrowing operation involves adding a single cell, adding at most $3^d$ supporting cells and \textbf{rebranding} at most $d^d$ cells per new root. So inserting a cell into a root takes at most $\log(\alpha)d^d(6d)^d$ time.

\begin{theorem}
\label{theorem:changes}
For each dynamic compressed quadtree $T_1$ over $\mathbb{R}^d$ we can maintain a smooth variant $T^*$ with at most:

\begin{itemize}
\item $\mathcal{O}((2d)^d)$ operations per split or merge on $T_1$.
\item $\mathcal{O}(\log(\alpha) d^2 (6d)^{d+1})$  operations per deletion or insertion of a compressed leaf.
\item $\mathcal{O}(d^2(6d)^d)$ time per upgrowing.
\end{itemize}
\end{theorem}

\section{Alignment.} 
\label{sec:alignment}
In the previous section we 
were treating our compressed quadtree as if 
for  each compressed component $A_1$, the cell of its root $R$ was {\em aligned} with the cell of the leaf $C_a$ that stores its compressed link.
That is, the cell of  $R$ can be obtained by repeatedly subdividing the cell of $C_a$. 
However, finding an aligned cell for the root of a new uncompressed component when it is inserted by 
the procedure insertCompressed is not supported in constant time in Real RAM model of computation. 
The reason is that the length of the compressed link, and thus the number of divisions necessary to compute the aligned root cell, may be arbitrary.

Instead of  
computing the aligned root cell at the insertion of each new compressed component, 
we allow compressed nodes to be associated with any hypercube of an appropriate size 
that is contained in the cell of the compressed ancestor $C_a$. This is a standard way to avoid the above problem~\cite{BLMM11,H11,LM12}.

\begin{figure}[h]
	\centering
	\includegraphics[width=270px]{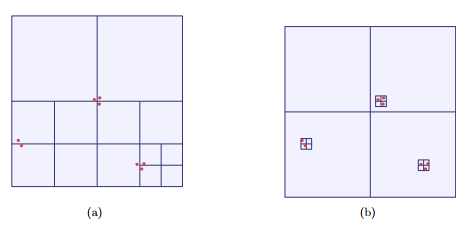}
\caption{An example of two compressed quadtrees on the same point set and slightly different root cells. The figure is borrowed from~\cite{LM12}. }
\label{fig:relocation}
\end{figure}

In our dynamic quadtree, after a number
of split operations on (the subtree of) $C_a$ and upgrowing operations on $A_1$, the size difference between the root of $A_1$ and its ancestor 
may become less than the compression constant $\alpha$, in which case we are supposed to \textbf{decompress} the compressed link, 
replacing it with the quadtree cells. 
This cannot be done if the root of $A_1$ and its ancestor 
are not aligned.  Thus we must maintain the \textbf{alignment property}, that is defined as follows. 

\begin{definition}
Let $A_1$ be  a compressed component of $T_1$ and $R$ be a hypercube containing the points stored in $A_1$. 
By \textbf{relocation} of $A_1$ to $R$, we mean building  the extended version of $A'^*$ (according to Definition~\ref{definition:T-star-final})
of the compressed quadtree  $A'_1$ such that $A'_1$  stores the same point set as $A_1$, and the root cell of  $A'_1$
is $R$.  
\end{definition}
  

\begin{definition}[Alignment property]
\label{def:align-invar}
Let $P$ be a point set in $\mathbb{R}^d$, let $T_1$ be an $\alpha$-compressed quadtree over $P$, and $T^*$ be the extended version of $T_1$ according to Definition~\ref{definition:T-star-final}. 
We define the \textbf{alignment property} for $T^*$ as follows: 
For any compressed link in $T^*$, if the length of the link is at most $4\log(\alpha)$, 
then the corresponding uncompressed component is aligned with the parent cell of the link.  
\end{definition}

In~the full version of \cite{lss13}, 
the alignment property 
is maintained in amortized additional $O(1)$ time (for $d$-dimensional quadtree with constant $d$), 
by  relocating the uncompressed component just before the decompression operation starts.  
Since the relocated quadtree may differ from the initial one a lot (see Figure~\ref{fig:relocation}),  
relocation might require time at least linear in the size of the component, this method clearly does not lead to a worst-case $O(1)$  bound.

We maintain the alignment property in worst-case $O(1)$ additional operations (for constant $d$) 
by  distributing the work to compute the aligned root cell
 and to relocate $A_1$, among merge and upgrowing operations that precede the corresponding decompression operation, 
and that reduce the size of the compressed link or increase the size of the uncompressed components.  
In particular, let $\ell$ be the length of the compressed link at the time it is created.  
We distribute the computation of the aligned root cell among first $\ell/4$ such operations, 
the work to relocate the uncompressed component among the second $\ell/4$ operations, 
and the necessary updates to the relocated component among the third $\ell/4$ operations. 
This way all the computations are finished by the time the length of the compressed link is at least $\ell/4$. 

The below Section~\ref{subsec:2d-align} describes our method in two dimensions and contains the main ideas; 
afterwards, in Section~\ref{subsec:dd-align} we generalize it to $d$ dimensions. 
 
\subsection{Alignment for 2-dimensional quadtrees.}
\label{subsec:2d-align}

Relocation of an uncompressed component can be done in time linear in the size of the component by a simple modification 
of the algorithm of~\cite[Theorem~3.7]{LM12}. 

\begin{theorem}
\label{thm:linear-align}
Let $\alpha$ be a sufficiently large constant and $P$ be a set of $m$ points in $\mathbb{R}^2$. 
Let $T_1$ be an  $\alpha$-compressed quadtree storing set $P$,
let $R$ be the root cell of $T_1$, and let $S$ be a square such  that $|S| \in \Theta(|R|)$ and $P \subset S$.
Then in $O(m)$ time we can relocate $T_1$ to cell $S$. 
\end{theorem}

\begin{proof}
Let $T'^*$ be the result of relocation of $T_1$ to cell $S$. 
Let  $\tilde{T}'$ be  the minimal 2-smooth version of $T'_1$. 
The algorithm of~\cite[Theorem~3.7]{LM12} in $O(m)$ time produces the tree  $\tilde{T}'$ out of $T_1$. 
Observe that the extended version
 $T'^*$ of $T'$ is a subset of $\tilde{T}'$. 
Since both $\tilde{T'}$ and $T'^*$ 
have $\mathcal{O}(m)$ cells, we can obtain $T'^*$ from $\tilde{T'}$ with $\mathcal{O}(m)$ merges. 
Therefore, we can accomplish our task by  running the algorithm of~\cite[Theorem~3.7]{LM12} for $T_1$ and $S$, 
and performing the necessary merges to go from the result of the algorithm, $\tilde{T}'$,  to 
the required tree $T'^*$.
\end{proof}

\begin{theorem}
\label{thm:const-alignment}
Let $P$ be a point set and $T_1$ be an $\alpha$-compressed quadtree over $P$ and 
$T^*$ be the extended variant of $T_1$ according to  Definition~\ref{definition:T-star-final}.  
The operations split and upgrowing on $T^*$ can be modified to 
maintain the alignment property
for $T^*$. The modified operations require same worst-case $O(1)$ time as the original ones.  
\end{theorem}

\begin{proof}
Consider any operation of insertion of a compressed leaf to $T_1$; let $A_1$ be the inserted uncompressed component,  and $C_a$
be the ancestor that stores the compressed link. 
Let $\ell$ be the length of this compressed link, that is, $|C_a| = 2^\ell|A_1|$. 

If $\ell \leq 4\log(\alpha)$, we compute  the properly aligned root cell for $A_1$  in $O(\log(\alpha))$ time,  by subdividing $C_a$ $\ell$ times. 
In what follows, we assume that $\ell > 4\alpha$. 
We now show how to maintain the alignment property
 for $A_1$ in that case. 

The following three events are important for us: 
(a) splitting the parent cell of the compressed link; 
(b) upgrowing of $A_1$; 
(c) splitting a leaf of $A_1$. 

Events (a) and (b) are the only ways to decrease the length of the compressed link of $A_1$. 
Events (b) and (c) are  the only ways to increase the size of the component, and in each of such events a constant number of cells is added to $A_1$.  
Our procedure performs the following three phases. 

\begin{itemize}
\item[Phase 1.] 
Compute the coordinates of the correctly aligned base cell for the points stored in $A_1$: 
every time one of the events (a)--(c) occurs, we perform $O(1)$ steps of determining the coordinates of the new base cell for $A$.  

\item[Phase 2.] Relocate $A_1$ to the root cell computed in Phase~1.  
We  perform  $O(1)$ steps of the relocation procedure every time one of the events (a)--(c) occurs. 
We are working with a local copy of $A_1$, ignoring the changes that may occur to $A_1$ during this phase.

\item[Phase 3.] Examine the changes to $A_1$ during Phase~2, and adjust the relocated variant $A'^*$ accordingly. 
We continue doing so until $A'^*$ is up-to-date with the current point set stored in $A_1$. To end the process on time, 
we adjust $A'^*$ at least  twice as fast as new changes to $A_1$ appear.    
 \end{itemize}

The exact number of steps performed at each event (a)--(c) during the Phase 1
 can be chosen so as to guarantee that the aligned root cell is computed after at most $\ell/4$ events.
 After Phase 1 is finished, the length of the compressed link is at least $3\ell/4$, and the size of $A_1$ is at most  $c_1\ell/4$ for some constant 
$c_1$. 

For the Phase 2, we are choosing the number of steps of relocation performed at
 each event (a)--(c) so that Phase 2 is completed after $\ell/4$ events.

Phase 3 performs several iterations, at every iteration it creates a local copy of $A_1$ and makes $A'^*$ compatible with the changes that have occured to the point set stored in $A_1$.
At the beginning of Phase 3, the number of split/merge operations needed to update $A'^*$ is at most $c_1\ell/4$. 
At the first iteration of Phase 3, we 
we perform these operations during $\ell/8$ events (a)--(c). 
The iterations of Phase~3 are holding the invariant that every iteration is 
at least twice shorter than the previous one, thus the entire Phase~3 is completed after at most $\ell/4$ events. 

The new tree $A'^*$ aligned with its parent cell is ready after at most $3\ell/4$ events. 
 The length of the compressed link at that time is at least $\ell/4$, which we assumed to be less than $\log(\alpha)$. 
Therefore, we are done in time. Observe finally  that each split or upgrowing operation in $T^*$ can be enhanced with additional operations by the above scheme 
at most two times.
This  completes the proof. 

\end{proof}

\subsection{Alignment for $d$-dimensional quadtrees.}
\label{subsec:dd-align}

In this section, we generalize the above result to $\mathbb{R}^d$. 
We first generalize Theorem~\ref{thm:linear-align}.

\begin{theorem}
\label{thm:dd-linear-align}
Let $\alpha$ be a sufficiently large constant with $\alpha > 2^{2^d}$,  and $P$ be a set of $m$ points in $\mathbb{R}^d$. 
Suppose we are given an $\alpha$-compressed quadtree $T_1$ for $P$, 
let $R$ be the root cell of $T_1$, and let $S$ be a hypercube such  that $|S| \in (3^{-d}|R|,3^d|R|)$ and $P \subset S$.
Then in $O(32^d\cdot m)$ time we can relocate $T_1$ to cell $S$. 
\end{theorem}

\begin{proof}
Analogously to the Theorem~\ref{thm:linear-align}, we run the algorithm of~\cite[Theorem~3.7]{LM12} for $T_1$ and $S$, 
and by applying the necessary $O((2d)^dm)$ merges  to the result,
we obtain the required extended version $T'^*$ of the $\alpha$-compressed quadtree 
storing the same point set as  $T_1$ and having $S$ as the root cell. Therefore, to prove the theorem it is enough to generalize the algorithm 
of~\cite[Theorem~3.7]{LM12} to $d$ dimensions.  

We first briefly summarize the algorithm. We call cells of the initial tree $T_1$ and of the partial relocated 2-smooth quadtree respectively 
 $T$- and $T'$-squares. 
The algorithm starts with identifying the $T'$-squares that intersect the root of $T_1$. They are obtained by constant number of splitting or upgrowing
of the partially relocated tree, which at that initial moment consists of the only square $S$. 
The leaves on the current partial relocated tree are called \emph{frontier} $T'$-squares.
The algorithm maintains the invariant that each frontier $T'$-square $C$
has a constant number of smallest $T$-squares of size at least $|C|$ that intersect $C$. 
Such $T$-squares are stored in the associated set of $C$; a frontier $T'$-square is 
called \emph{active} if is has non-empty associated set. 
The \emph{main body} of the algorithm repeatedly splits active frontier $T'$-squares
 and computes the associated sets of their children. After each split, the 2-smoothness of the partial relocated tree is restored. 
The \emph{secondary stage} builds $O(\log{\alpha})$-height small compressed quadtrees to separate the points and the large enough uncompressed components that are still not separated after all active $T'$-squares are processed.   Finally the recursive calls are set to process the remaining uncompressed components. 

Although the algorithm is formulated for $\mathbb{R}^2$, it can be directly applied to $d$-dimensional quadtrees.
In what follows we state the dependencies on $d$. 
In the initialization step of the algorithm the number of $T'$-squares is $O(3^d)$. 
The invariant maintained in the main body of the algorithm is that each frontier square of the partially relocated tree has 
$O(2^d)$ associated $T$-squares. 
The time complexity of one split of a $T'$-square $C$ is $O(4^d)$, 
since each of the $O(2^d)$ squares of the associated set of $C$ is compared against each of the $2^d$ 
children of $C$ 
Similarly, each split operation during the \emph{secondary stage} takes $O(4^d)$ time.  
The step called \emph{setting up the recursive call} produces a set $X$ of at most $2^d$ $T'$-squares, 
computing a pair of base squares for the recursive call takes $O(\alpha)$ splits, that is $O(\alpha\cdot 2^d)$ time. 
Lemma 3.8 of~\cite{LM12} provides a charging scheme for splits in the main body to active $T'$-squares.
The number of times one square can be charged is $O(8^d)$. Indeed, same active square is charged (a) at most once when it is split; 
(b) for the balance splits, it is charged $O(2^d)$ times as an active square five levels above the neighbor $N'$ of the square being split, such that   
$N'$ triggers the split; and (c) same type of charging as in (b), but the descendant square being the $i$-th neighbor of $N'$ ($i$ is at most $4^d$) - $8^d$ times.
Lemma~3.9 of~\cite{LM12} transforms now into the statement that any subtree produced by the secondary stage has height $O(2^d\log{\alpha})$.  
Indeed, the size of the secondary associated list is at most $2^d$, any subset of it may appear at most twice at the same level of the subtree, and after a subset appears it stays for $O(\log \alpha)$
levels, and the associated sets of each descendant of $S'$ are subsets of the associated set of $S'$; thus the maximum depth of any leaf of the tree is $O(2^d\log{\alpha})$.  
If we choose $\alpha > 2^{2^d}$, the argument similar to the one of Lemma~3.9 implies the claim. 
\end{proof}

We now are ready to generalize the algorithm of Theorem~\ref{thm:const-alignment} to $d$ dimensions. 
We follow exactly the same scheme as in Theorem~\ref{thm:const-alignment}. 
In Phase~1, the aligned root cell is computed by performing $2^d\ell$ divisions during $\ell/4$ events, i.e., in $O(2^d)$ time per event. 
In Phase~2, the relocation of the uncompressed component is performed, due to Theorem~\ref{thm:dd-linear-align}, in $O(32^d\ell)$ time during another  $\ell/4$ events, 
 i.e., in $O(32^d)$ time per event. 
In Phase~3, the adjusting is performed, due to Theorem~\ref{theorem:changes},  in $O(d^2 \cdot (6d)^{d}\ell)$ time, during the third $\ell/4$ events, i.e., 
$O(d^2 \cdot (6d)^{d})$ per event.    
Note that applying Theorem~\ref{thm:dd-linear-align} imposes an additional constraint on $\alpha$ to be at least $2^{2^d}$. 
This results in the following.
\begin{theorem}
\label{thm:dd-const-alignment}
Let $P$ be a point set in $\mathbb{R}^d$ and $T_1$ be an $\alpha$-compressed quadtree over $P$ for a sufficiently big constant $\alpha$ with $\alpha > 2^{2^d}$ and 
$T^*$ be the extended variant of $T_1$ according to  Definition~\ref{definition:T-star-final}.  
The operations split and upgrowing on $T^*$ can be modified to 
maintain the alignment property 
for $T^*$. The modified operations require worst-case $O(32^d + d^2 \cdot (6d)^{d})$ time. 
\end{theorem}

\section{Applications.}
\label{sec:applications}

Many publications in computational geometry use a concept which we shall dub \textbf{principal neighbor access}: The idea that for any cell $C$ we can find its relevant neighbors in constant time:
\begin{definition}
Given a quadtree $T$ over $\mathbb{R}^d$, we say that we have \textbf{principal neighbor access} if for any cell $C$ in $T$ we can find the smallest cells $C'$ in $T$ with $|C'| \ge |C|$ and $C'$ neighboring $C$ in constant time if $d$ is constant.
\end{definition}
Bennet \and Yap in \cite{BY17} implement principal neighbor access by storing explicit \textbf{principle neighbor pointers} to the larger neighbors $C'$. Khramtcova and L\"offler in  \cite{khramtcova2017dynamic} achieve principal neighbor access with the well known \textbf{level pointers} on a smooth quadtree. These two unique ways to guarantee principle neighbor access were also noted by Unnikrishnan \etal in \cite{unnikrishnan1988threaded} where a \emph{threaded} quadtree in \cite{unnikrishnan1988threaded} maintained the equivalent to \textbf{principle neighbor pointers} as opposed to a \emph{roped} quadtree in \cite{samet1982neighbor} which only maintained level pointers.

Principle neighborhood access allows us to traverse the neighborhood of any cell $C$ in constant time. Bennet and Yap observe that any (non-compressed) quadtree must be smooth to dynamically maintain level pointers by using a sequence of cells that they insert.  In Figure \ref{fig:levelpointer} we show their example and our own example where we show that if quadtrees are compressed, you even need $\Theta(n)$ time for a single split operation to update the level pointers. For this reason Bennet and Yap develop their amortized-constant dynamic smooth quadtrees in $\mathbb{R}^d$ in \cite{BY17}. We note that most applications that use principal neighbor access (dynamic variant of collision detection \cite{mezger2003hierarchical}, ray tracing \cite{aronov2006cost, macdonald1990heuristics} and planar point location \cite{khramtcova2017dynamic, lss13}) often run many operations parallelized on the GPU. In such an environment amortized analysis can become troublesome since there is a high probability that at least one GPU-thread obtains the worst-case $\mathcal{O}(n)$ running time. In that scenario the other threads have to wait for the slow thread to finish so the computations effectively run in $\mathcal{O}(n)$ time which makes our worst-case constant time algorithm a vast improvement.

\begin{figure}[H]
	\centering
	\includegraphics[width=300px]{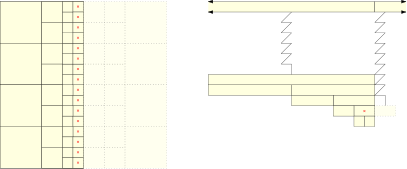}
	\caption{(left) A quadtree in $\mathbb{R}^2$ in white opaque cells with depth $k$. If we want to add the dotted cells with regular splits then with every split all the cells that contain a cell with a red cross might need a new level pointer. This means that we have to adjust $\sum_i^{k-1} 2^i 2^{k-1-i} = \mathcal{O}(n\log n )$ pointers even though we perform $\mathcal{O}(n)$ split operations. (right) A compressed quadtree in $\mathbb{R}^1$. Here we it is evident that even a single insertion can demand $\mathcal{O}(\log n )$ time to update the level pointers. If we want to add the dotted cell, we need to find the cell marked with a red cross. This can best case be done with binary search which takes $\Theta(\log n)$ time. }
	\label{fig:levelpointer}
\end{figure}

Many publications either directly or indirectly make use of level pointers. In this subsection we mention three distinguished areas in computational geometry and we discuss recent publications in those fields that directly benefit from our constant-time dynamic compressed quadtrees.

\subsubsection*{Ray shooting and planar point location.}

Two related well studied topics in computational geometry are ray shooting and planar point location. The ray shooting problem arises in many different contexts and is a bottleneck of ray tracing in computer graphics. In the ray shooting problem we are given a set $\Gamma$ of objects in a hyperplane and a halfline in the hyperplane (often $\mathbb{R}^2$ or $\mathbb{R}^3$) and we ask which of the objects (if any) the ray hits first. In the planar point location, we store a collection $\Gamma$ of 
objects in a hyperplane, and ask for an arbitrary point in the plane: which regions intersect this point. We can see ray tracing as tracing a line in the plane and stabbing queries as tracing a line perpendicular to the plane. In both problems we can store the objects in a multi-dimensional compressed quadtree \cite{aronov2006cost, macdonald1990heuristics, lss13, khramtcova2017dynamic}. If we want the problem to work on dynamic input we need to have dynamic compressed quadtrees. Moreover, each of the objects can transform independently which makes the problem highly suited for parallelization which amplifies the need for worst-case algorithms over amortized algorithms. \\

In \cite{macdonald1990heuristics} Macdonald \etal develop heuristics for ray tracing objects in a dynamic setting and they mention dynamic quadtrees in $\mathbb{R}^3$. In \cite{aronov2006cost} Aronov \etal analyze the running time of ray tracing and  use level pointers in three dimensional quadtrees. Their quadtree is not dynamic but they need neighbor pointers to access the neighborhood of a cell. If anyone would want to extend their research to work in a dynamic setting they would need our quadtrees. In \cite{lss13} L\"{o}ffler \etal and in \cite{khramtcova2017dynamic} Khramtcova and L\"{o}ffler want to dynamically maintain a moving set of points or regions in $\mathbb{R}^{1,2}$ in Real RAM subject to \textbf{stabbing queries}. To perform the sub-logarithmic updates and the $\mathcal{O}(\log(n) + k)$ stabbing queries they need to access the neighborhood of each cell in $\mathcal{O}(1)$ time. Moreover to maintain the dynamic compressed quadtree in Real RAM they also need our alignment. In the full version of \cite{lss13} the authors hint at the use of level pointers and note that we might need smooth quadtrees to obtain them. In \cite{khramtcova2017dynamic} Khramtcova and L\"{o}ffler conjecture that these pointers might also be obtainable without smooth quadtrees but the lower bound proofs in Figure \ref{fig:levelpointer} disprove that conjecture.

\subsubsection*{Voronoi diagrams and quadtrees.}

It is well known that the Delaunay triangulation of a point set is the dual graph of the Voronoi diagram of that point set \cite{de2008computational}. In \cite{LM12} L\"{o}ffler \etal show that given a compressed quadtree of a point set of $n$ points, we can construct the Delaunay triangulation of that point set in $\mathcal{O}(n)$ time instead of $\mathcal{O}(n \log n)$ time and vice versa. The idea is that the quadtree decomposition provides information about the Delaunay triangulation.  In  \cite{devillers1992fully} Devillers \etal dynamically maintain the Delaunay triangulation on a set of points in $\mathcal{O}(\log n )$ expected time per insertion and deletion. They use a dynamic quadtree of the point set and use the work of Bennet \etal in \cite{BY17} to use level pointers. With our worst-case constant smooth quadtrees their running time is still expected $\mathcal{O}(\log n )$ per operation but the uncertainty in running time now only comes from how the Delaunay triangulation is shaped and not from the dynamic quadtree. Moreover we extend their results from working on dynamic bounded-spread point sets to arbitrary point sets and on a Real RAM pointer machine.

Bennet \etal in \cite{bennett2016planar} work on a problem where given a set of scalar functions from $\mathbb{R}^2$ to $\mathbb{R}$, they are interested in the minimization diagram of the set of functions. They use (compressed) quadtrees and level pointers to construct a Voronoi diagram on the domain which they use to compute the minimization diagram. They used their work in \cite{BY17} to perform the construction in amortized constant time. With our results this becomes worst-case constant.

\subsection*{Other data structures that use quadtrees.}

As a last application we mention other data structures that dynamic compressed quadtrees as a basis.

In \cite{de2012kinetic} de Berg \etal want to maintain a kinetic compressed quadtree in what they call the Black-Box model and they mention applications in collision detection. Their work also requires them to access the neighborhood of a cell and they indirectly make use of level pointers. Therefore their work requires our work for the implementation of level pointers, moreover if they want their data structure to work in real RAM they also need our alignment.

Lastly Mount \etal in \cite{park2012self} want to maintain a self-adjusting data structure for multidimensional point sets. They describe a \emph{self-adjusting} data structure as "a data structure that dynamically reorganizes itself to fit the pattern of data accesses." and mention \textbf{splay trees} as a common implementation. They use a compressed balanced quadtree which they maintain in amortized constant time per operation to update and re-balance their splay tree in amortized $\mathcal{O}(\log n )$ time. We believe that with our dynamic quadtree their work could be improved to worst-case $\mathcal{O}(\log n )$ running time. Moreover just like with \cite{de2012kinetic} we note that if they want their data structure to work in the Real RAM model then they need our algorithm for alignment.


\small
\section*{Acknowledgements.}
The authors would like to thank Joe Simons and Darren Strash for their inspiring initial discussion of the problem.
I.H. and M.L. were partially supported by the Netherlands Organisation for Scientific Research (NWO) through project no 614.001.504.
E.K. was partially supported by the SNF Early Postdoc Mobility grant P2TIP2-168563, Switzerland, and F.R.S.-FNRS, Belgium.

\bibliography{dbq}

\appendix

\section{Our work and the work of L\"{o}ffler, Simons, Bennet and Yap.}
\label{sec:bennet}

Since we only focus on operation on the quadtrees mentioned in Section \ref{sec:preliminaries} and forget the time it takes to locate points, we view a quadtree as a self-contained concept: An uncompressed quadtree is simply a tree of $n$ nodes subject to the \textbf{split} and \textbf{merge} operation and we want that quadtree to be smooth. \\
Simons and L\"{o}ffler proposed to make a distinction between the original cells of the quadtree $T$ (called the \textbf{true} cells), and the quadtree cells that are added on top of $T$ to maintain balance (\textbf{balancing cells}). \emph{true} cells must be $2$-smooth (or smooth) whereas the cells used for balance can be $4$-smooth. The result would be a smooth quadtree where cells are $4$-smooth. The split and merge operation can only be used on \emph{true} cells and they tried to dynamically split and merge balancing cells such that the resulting tree would still be smooth. Their approach did not work in $\mathbb{R}^d$ with $d \ge 2$ however, for a similar reason why the approach of Bennet and Yap has to be amortized.\\
Bennet and Yap used a slightly different approach from Simons and L\"{o}ffler. In their case, all the cells in $T$ are true cells before the split operation occurs. When they split a cell, they apply an extra operation which they call \textbf{ssplit} (smooth split) which adds extra cells that makes all cells in the quadtree $2$-smooth. The difference with the approach of Simons and L\"{o}ffler is that after their operation, they can invoke the split operation on \emph{any} cell in their quadtree instead of only on the \emph{true} cells. They show that if you want all cells of the quadtree to be $m$-smooth with the same constant $m$, that their amortized bounds are tight.

\subsubsection*{Intuitively why the approaches of L\"{o}ffler and Bennet must be amortized.}

Allow us to intuitively explain why both approaches can only work with amortized bounds in $\mathbb{R}^2$ and higher: Amortized bounds are analyzed over a sequence of possible splits. So intuitively we can imagine that for Bennet and Yap, the cells that actually need to be split create \emph{true} cells and that all other cells are just added for balance and are thus \textbf{balance} cells. In Figure \ref{fig:4balanceex}, let the white cells be the true cells and the green cells be the balancing cells. We made them 4-smooth as an example for the work of Simons and L\"{o}ffler but a similar construction is possible for the work of Bennet and Yap.

\begin{figure}[H]
	\centering
	\includegraphics[width=300px]{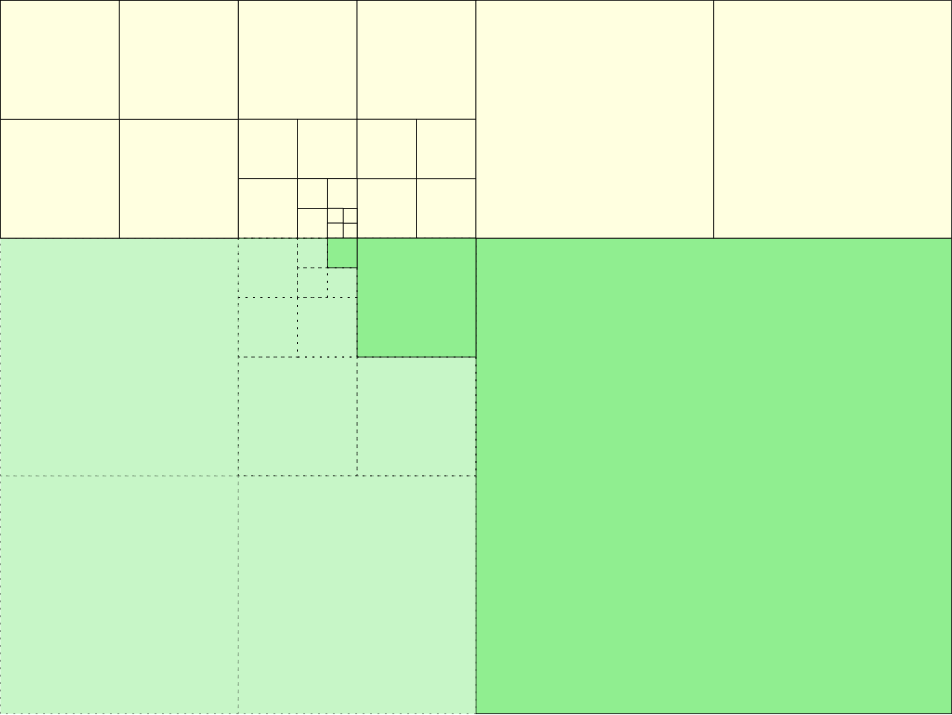}
	\caption{A set of true cells that store a point set (white) and its balancing cells (green).}
	\label{fig:4balanceex}
\end{figure}

In Figure \ref{fig:4balanceex} shows part of a quadtree where the true cells created 3 consecutively neighboring green cells which we left opaque. Each of consecutive cells is exactly its balancing constant larger than the previous cell. Splitting the smallest cell, now creates a \textbf{cascading} effect where we have to consecutively split all the cells in the chain to maintain smoothness. Bennet and Yap call such a sequence of neighboring cells a \textbf{forcing chain}. You can repeat this pattern to create arbitrary long rectilinear paths throughout your bounding box. Setting up a chain of length $k$ intuitively takes $\mathcal{O}(k)$ time which allowed Bennet and Yap to create insertions for quadtrees in $\mathbb{R}^d$ which amortized take  $\mathcal{O}(2^d(d+1)!)$ time. 

\section{A naive approach on the dynamic algorithm to maintain smooth quadtrees.}
\label{app:naive}
We know that we should be able to support split and merge operations in the true set $T_1$ in $\mathcal{O}(d^d)$ split and merge operations, but we do not have an algorithm for it yet. Our first attempt is to simply add an extra step to the traditional split operator: After we have split a cell with brand $j$, we invoke a new method called $AfterSplit(C)$ on all the new children. This method checks the size of all the $d$ neighbors of $C$, and splits every too large cell into neighbors with a brand $j+1$.

\begin{algorithm}[h]
\caption{The procedure for after splitting a cell.}
\begin{algorithmic}[1]
\Procedure{AfterSplit(Cell C, Balance j)}{}
\For{Cell N $\in$ Neighbors}
\If{$\frac{|N|}{|C|} > 2^j$}
\State $V \gets Split(N, j+1)$
\For{Cell C' $\in V$}
\State $AfterSplit(C', V)$
\EndFor
\EndIf
\EndFor
\EndProcedure
\end{algorithmic}
\end{algorithm}

\begin{figure}[H]
	\centering
	\includegraphics[width=300px]{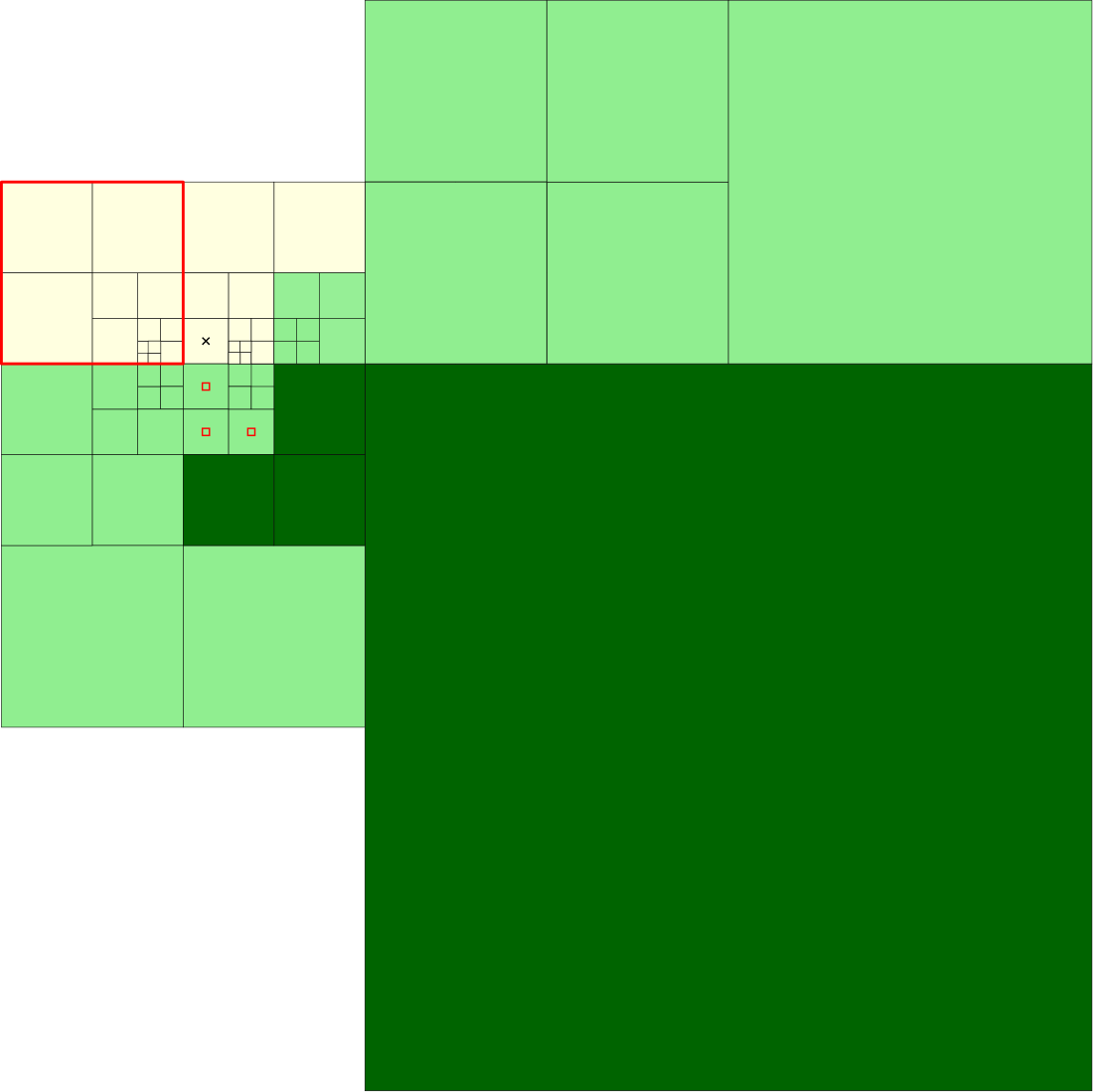}
	\caption{A quadtree in $\mathbb{R}^2$ with $T_1$ in white, $T_2$ in light green and $T_3$ in dark green}
	\label{fig:gammaorder}
\end{figure}

\begin{lemma}
The tree computed by Algorithm 1 does not necessarily satisfy Definition \ref{def:tree}.
\end{lemma}

\begin{proof}
We show this with an example shown in Figure \ref{fig:gammaorder}.
In the figure the cells in $T_1$ are white, in $T_2$ are light green and cells in $T_3$ are dark green. Assume that we start with a single large true cell, the bounding box, which is significantly larger than the figure. And that we create the true cells in the red square first using our split operation and after each split we invoke the AfterSplit procedure. \\

Splitting the cells in the red box will then create the light green cells of brand 2 below, and eventually five dark green cells in $T_3$. But the cells starting from the red cells are not created yet. \\

Now assume that we create the true cells adjacent to the red square with our split operation and after each split we invoke the AfterSplit procedure. Eventually we will have to create the red cells with brand 2, on top of the cells with brand 3. Given this static snapshot of $T_1$, we can clearly see that our $T^*$ does not adhere to Definition \ref{def:tree}: $T_2$ should be the minimal set of cells that makes all the cells in $T_1$ 2-smooth. However, the cell with the cross is clearly not 2-smooth without additional cells. But the adjacent cells are cells in $T_3$ and not in $T_2$.

\end{proof}

If our algorithm creates a tree $T^*$ which does not implement Definition \ref{def:tree}, Lemma  \ref{lemma:unique}, \ref{lemma:branding} and Lemma \ref{lemma:brandingprinciple} do not have to hold. That implies that we have no guarantee that our resulting quadtree has smooth cells in the last layer of balancing cells and that our algorithm might take more than $\mathcal{O}(d^d)$ split and merge operations.

\section{An example of why with compressed extended quadtrees, we need a new definition.}
\label{appendix:compression}

In a dynamic non-smooth compressed quadtree $T_1$ we can clearly insert and remove a compressed root in  constant time even though the virtual tree of $T_1$ and the virtual tree of the new quadtree $T_1'$ might differ in more than $\mathcal{O}(n)$ cells. However if we define the extended quadtrees $T^*$ and $T'^{*}$ of $T_1$ and $T_1'$ according to Definition \ref{def:tree} the two extended trees might differ in more than a constant number of cells even in their actual stored variants. We show this with the following example:

\begin{figure}[H]
	\centering
	\includegraphics[width=400px]{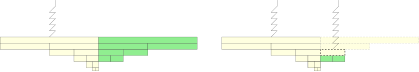}
	\caption{On the left we see a quadtree of true cells $T_1$ in $\mathbb{R}^1$ in white and $T^*$ with cells of brand 2 in green. On the right we have the tree $T_1'$ which added a compressed root $R$ next to the quadtree, the cell with the dashed lines. $T_1$ and $T_1'$ clearly differ in only one cell in storage. The cells greater than the new compressed root $C$ in $T_1'$ are already smooth by the cells in the compressed link. So possibly $\Theta(n)$ light green cells in $T^*$ should be removed and we see that $T^*$ and $T'^{*}$ hence differ in more than a constant amount of cells, even in storage. }
	\label{fig:compressionexample}
\end{figure}

\end{document}